\documentclass[11pt]{article}
\usepackage{lmodern}
\usepackage{natbib}
\usepackage[a4paper, total={6.4in, 9in}]{geometry}
\usepackage{amsmath,amsthm,amssymb}
\newtheorem{definition}{Definition}
\newtheorem{lemma}{Lemma}
\newtheorem{theorem}{Theorem}

\newtheorem{example}{Example}
\AtBeginDocument{}
\usepackage{sectsty}
\allsectionsfont{\centering}
\usepackage{titlesec}
\titlelabel{\thetitle.\quad}
\usepackage{xurl}
\usepackage{eqnarray}
\usepackage{filecontents}
\usepackage{adjustbox}
\usepackage{xurl}
\usepackage{color}

\usepackage{algorithm,algpseudocode}
\usepackage{authblk}
\newcommand{\setalglineno}[1]{%
	\setcounter{ALG@line}{\numexpr#1-1}}
\newcommand*{\angela}[1]{\textcolor{black}{#1}}
\newcommand*{\gramm}[1]{\textcolor{black}{#1}}
\usepackage{xr}

\title{Procrustes analysis for high-dimensional data}

\author[1]{Angela Andreella}
\affil[1]{Department of Statistical Sciences, University of Padova, Italy}

\author[2]{Livio Finos}
\affil[2]{Department of Developmental Psychology and Socialization, University of Padova}
\date{}

\begin{document}












\maketitle
\begin{abstract}
The Procrustes-based perturbation model \citep{Goodall} allows minimization of the Frobenius distance between matrices by similarity transformation. However, it suffers from non-identifiability, critical interpretation of the transformed matrices, and inapplicability in high-dimensional data. We provide an extension of the perturbation model focused on the high-dimensional data framework, called the ProMises (Procrustes von Mises-Fisher) model. The ill-posed and interpretability problems are solved by imposing a proper prior distribution for the orthogonal matrix parameter \gramm{(}i.e., the von Mises-Fisher distribution\gramm{)} which is a conjugate prior, resulting in a fast estimation process. Furthermore, we present the Efficient ProMises model for the high-dimensional framework, useful in neuroimaging, where the problem has much more than three dimensions. We found a great improvement in functional \gramm{m}agnetic \gramm{r}esonance \gramm{i}maging \gramm{(fMRI)} connectivity analysis \gramm{because} the ProMises model permits incorporat\gramm{ion of} topological brain information in the alignment's estimation process.

\end{abstract}


\section{Introduction}\label{introduction}
The Procrustes problem \gramm{is aimed at matching} matrices using similarity transformations by minimizing their Frobenius distance. It allows compar\gramm{ison of} matrices with dimensions defined in an arbitrary coordinate system.
This method raised the interest of applied researchers hence highlighting its potentiality through a plethora of applications in several fields, such as ecology \citep{saito2015should}, biology \citep{rohlf1990extensions},  analytical chemometrics \citep{ANDRADE2004123}, and psychometrics \citep{Green,psychom}.

The interest of a large audience from applied fields stimulates, in parallel, the growth of a vast \gramm{body of} literature. Despite this, essentially all applications comprise spatial coordinates \gramm{(}i.e., two\gramm{-} or three-dimensional\gramm{)}. \cite{Haxby} first introduced the use of this approach into a different context: align functional Magnetic Resonance Images (fMRI). The coordinates are hence substituted by voxels \gramm{(}i.e., three-dimensional pixels\gramm{)}, and the problem becomes inherently high-dimensional. The approach rapidly grew in popularity in the neuroimaging community because of its effectiveness. However, the proposed solution is naive; the extension from the spatial context to a more general and high\gramm{-}dimensional one \gramm{is} a theoretical challenge that needs adequate attention.

The most serious concern is the results' interpretability. In most cases\gramm{,} Procrustes methods turn into an ill-posed problem. It is a barely noticeable problem with spatial coordinates \gramm{because} the solution is unique up to rotations hence\gramm{,} the user has the freedom to choose the point of view that provides the nicest picture. When the dimensions do not have a spatial meaning, any rotation completely changes the interpretation of the results. 

To tackle \gramm{this} problem, we revise the perturbation model \citep{Goodall}, which rephrases the Procrustes problem as a statistical model. The matrices are defined as a random perturbation of a reference matrix plus an error term. The perturbation is expressed by rotation, scaling, and translation and the matrix normal distribution \citep{Gupta} is assumed for the error terms. \gramm{Like} \cite{MardiaBayesian, MardiaBayesian1}, we assume that the orthogonal matrix parameter follows the von Mises-Fisher distribution. We prove that the proposed prior distribution is conjugate, making the estimation process quite fast. Indeed, the maximum a posterior estimate of the orthogonal matrix parameters results in a minor modification of the original solution \gramm{because} the prior information enters into the pairwise cross-product of the matrices to be aligned. 
The prior distribution plays the role of regularizing term, resolving the non-identifiability of the orthogonal matrix parameter. In the application to fMRI data in Section \ref{application}, we further show that specification of a prior \gramm{distribution} permits the integration of functional and topological aspects, which largely improve\gramm{s} the \gramm{results'} interpretability. We then propose a comprehensive approach to the Procrustes problem: the ProMises (Procrustes von Mises-Fisher) model. 

The second problem raised by the extension to the high-dimensional framework is computational. The estimation algorithm of a Procrustes-based method involves a series of singular value decompositions of $m\times m$ matrices ($m$ dimensions). In a typical fMRI data set, the subjects \gramm{(}i.e., the matrices to be aligned\gramm{)} have a few hundred \gramm{(}observations/rows\gramm{)} $n$ and hundred\gramm{s of} thousands \gramm{of} voxels \gramm{(}dimensions/columns\gramm{)} $m$. We prove that the minimization problem can be solved by a series of singular value decompositions of $n\times n$ matrices, reducing the computation burden and making the ProMises model applicable to matrices with virtually any number of  columns. \angela{We denote this approach as the Efficient ProMises model.}

\angela{We emphasize here that the problem of aligning fMRI data is not three-dimensional as it could appear at first \gramm{glance}, \gramm{although} it is high-dimensional: \gramm{E}ach voxel is one dimension of the problem.
Nevertheless, in such conditions, three critical issues arise\gramm{:} The first is that \gramm{the} Procrustes method combines any voxel inside the brain without distinguishing between adjacent and distant anatomical locations. This can be questionable \gramm{because} the voxels have a spatial organization, and, despite the inter-subject variability, we expect some degree of spatial similarity between subjects in their functional organization. Therefore\gramm{,} we want the subjects' voxels of a given location \gramm{to be more likely to contribute to} to the construction of voxels with the same location in the common space.
The second issue revolves around the non-identifiability of orthogonal transformations $\boldsymbol{R}_i$, where $i$ indexes the subjects. The Procrustes method does not return a unique solution of the maximum likelihood estimate for $\boldsymbol{R}_i$: \gramm{G}iven a solution \gramm{(}i.e., a three-dimensional image\gramm{)}, any linear combination that mixes the voxels' values is an equivalent solution to the problem. The solutions are equivalent only from a mathematical point of view \gramm{because} the practical consequence is the loss of \gramm{results'} topological interpretability.
The third issue is the computational load: applying the Procrustes-based alignment to the whole brain implies the decomposition of many square matrices of dimensions roughly equal to $200.000$ \gramm{(}i.e., the number of voxels.\gramm{)}}

\angela{The Efficient ProMises model resolves all these three issues. The use of a properly chosen prior shrinks the estimate to the anatomical solution (i.e., no rotation), hence making the solution unique and interpretable from an anatomical point of view. Finally, as mentioned before, the Efficient implementation permits perform\gramm{ing} functional alignment on high-dimensional data such as fMRI data.}

The paper is organized as follows. Section \ref{background} introduces the perturbation model \citep{Goodall}, stressing its critical issues. Section \ref{procrustes-solution-with-regularization} illustrates the ProMises model and its challenges: identifiability, interpretability, and flexibility. Section \ref{light} define\gramm{s} the Efficient version of the ProMises model, which permits appl\gramm{ication of} the functional alignment \gramm{to} high-dimensional data. Finally, the presented model is evaluated by analyzing task-related fMRI data in Section \ref{application}. The proposed method is developed in \url{https://github.com/angeella/ProMisesModel} using Python \citep{van1995python}'s programming language, and in \url{https://github.com/angeella/alignProMises} using the \texttt{R} \citep{R} programming language. \angela{We report the proofs of the main lemmas here, whereas the remaining proofs are included in the appendix.}

\section{Perturbation Model}\label{background}
\subsection{Background}
Let $\{\boldsymbol{X}_{i}  \in {\rm I\!R}^{n \times m} \}_{i = 1,\dots,N}$ be a set of matrices to be aligned. The Procrustes-based method uses similarity transformations to match \gramm{each} matric\gramm{x} to the target one as close\gramm{ly} as possible, according to the Frobenius distance.

Each matrix $\boldsymbol{X}_i$ could be then assumed to be a similarity transformation of a shared matrix $\boldsymbol{M} \in {\rm I\!R}^{n \times m}$, which contains the common reference space's coordinates, plus a random error matrix $\boldsymbol{E}_i \in {\rm I\!R}^{n \times m}$. The perturbation model proposed by \cite{Goodall} is then reported. 

\begin{definition}[\cite{Goodall}]\label{Perturbation}
	Let $\{\boldsymbol{X}_{i}  \in {\rm I\!R}^{n \times m} \}_{i = 1,\dots,N}$ be a set of matrices to be aligned, and $\mathcal{O}(m)$ the orthogonal group in dimension $m$. The perturbation model is defined as
	\begin{equation*}
		\boldsymbol{X}_i= \alpha_i (\boldsymbol{M} + \boldsymbol{E}_i)\boldsymbol{R}_i^\top + \mathbf{1}_n^\top \boldsymbol{t}_i \quad\quad \text{subject to }\quad \boldsymbol{R}_i \in \mathcal{O}(m),
	\end{equation*}
	where $\boldsymbol{E}_i \sim \mathcal{MN}_{n,m}(0,\boldsymbol{\Sigma}_n,\boldsymbol{\Sigma}_m)$ \gramm{---}i.e., the matrix normal distribution with $\boldsymbol{\Sigma}_n \in {\rm I\!R}^{n \times n}$ and $\boldsymbol{\Sigma}_m \in {\rm I\!R}^{m \times m}$ scale parameters\gramm{---} $\boldsymbol{M} \in {\rm I\!R}^{n \times m}$ is the shared matrix, $\alpha_i \in {\rm I\!R}^{+}$ is the isotropic scaling, $\boldsymbol{t}_i \in {\rm I\!R}^{1 \times m}$ defines the translation vector, and $\mathbf{1}_n \in {\rm I\!R}^{1 \times n}$ is a vector of ones.
\end{definition}

To simplify the problem, the column-centered $\boldsymbol{X}_i$ are considered. The distribution is
\begin{equation*}
	\text{vec}(\boldsymbol{C}_n\boldsymbol{X}_i| \boldsymbol{R}_i, \alpha_i, \boldsymbol{M}, \boldsymbol{\Sigma}_m, \boldsymbol{\Sigma}_n) \sim \mathcal{N}_{n m}(\text{vec}(\alpha_i \boldsymbol{C}_n\boldsymbol{M} \boldsymbol{R}_i^\top),  \boldsymbol{R}_i \boldsymbol{\Sigma}_m \boldsymbol{R}_i^\top \otimes \alpha_i^2 \boldsymbol{C}_n \boldsymbol{\Sigma}_n \boldsymbol{C}_n^\top),
\end{equation*} 
where $\boldsymbol{C}_n = \boldsymbol{I}_n - \frac{1}{n} \boldsymbol{J}_n$, $\boldsymbol{I}_n \in {\rm I\!R}^{n \times n}$ is the identity matrix, $\boldsymbol{J}_n$ is a $n \times n$ matrix of ones, and $\text{vec}(\boldsymbol{A})$ the vectorization of the matrix $\boldsymbol{A}$. Let the singular value decomposition of $\boldsymbol{C}_n = \boldsymbol{\Gamma} \boldsymbol{\Delta} \boldsymbol{\Gamma}^\top$, where $\boldsymbol{\Gamma} \in {\rm I\!R}^{n \times (n-1)}$, then\gramm{:}
\begin{align*}
	\text{vec}(\boldsymbol{\Gamma}^\top \boldsymbol{C}_n\boldsymbol{X}_i| \boldsymbol{R}_i, \alpha_i, \boldsymbol{M}, \boldsymbol{\Sigma}_m, \boldsymbol{\Sigma}_n) \sim& \mathcal{N}_{(n-1) m}(\text{vec}(\alpha_i \boldsymbol{\Gamma}^\top \boldsymbol{C}_n\boldsymbol{M} \boldsymbol{R}_i^\top),   \boldsymbol{R}_i \boldsymbol{\Sigma}_m \boldsymbol{R}_i^\top \otimes \alpha_i^2 \boldsymbol{\Gamma}^\top \boldsymbol{C}_n \boldsymbol{\Sigma}_n \boldsymbol{C}_n^\top \boldsymbol{\Gamma}).
\end{align*}

The $\boldsymbol{\Gamma}^\top$ transformation leads to independence between the rows of $\boldsymbol{E}_i$ without affecting the estimation of $\boldsymbol{R}_i$. $\boldsymbol{\Gamma}^\top \boldsymbol{C}_n \boldsymbol{X}_i$ and $\boldsymbol{\Gamma}^\top \boldsymbol{C}_n \boldsymbol{M}$ have now $n-1$ rows; however, we can simply re-project them on ${\rm I\!R}^{n \times m}$ using $\boldsymbol{\Gamma}$. \gramm{Because} we can always write $\tilde{\boldsymbol{X}}_i = \boldsymbol{\Gamma} \boldsymbol{\Gamma}^\top \boldsymbol{C}_n \boldsymbol{X}_i$,  $\tilde{\boldsymbol{M}} = \boldsymbol{\Gamma} \boldsymbol{\Gamma}^\top \boldsymbol{C}_n \boldsymbol{M}$, \gramm{and} $\tilde{\boldsymbol{\Sigma}}_n = \boldsymbol{\Gamma}\boldsymbol{\Gamma}^\top \boldsymbol{C}_n \boldsymbol{\Sigma}_n \boldsymbol{C}_n^\top \boldsymbol{\Gamma}\boldsymbol{\Gamma}^\top$ without loss of generality\gramm{,} we can re-write Definition \ref{Perturbation} as \gramm{follows:}
\begin{definition}\label{Procrustes_new}
	\begin{align}
		\boldsymbol{X}_i= \alpha_i (\boldsymbol{M} + \boldsymbol{E}_i)\boldsymbol{R}_i^\top \quad\quad \text{subject to }\quad \boldsymbol{R}_i \in \mathcal{O}(m)
		\label{eq:eq15_new},
	\end{align}
	where $\boldsymbol{E}_i \sim \mathcal{MN}_{n m}(0,\boldsymbol{\Sigma}_n,\boldsymbol{\Sigma}_m)$. This way, we obtain \gramm{the following:}
	\begin{align}
		\text{vec}(\boldsymbol{X}_i| \boldsymbol{R}_i, \alpha_i, \boldsymbol{M}, \boldsymbol{\Sigma}_m, \boldsymbol{\Sigma}_n) \sim \mathcal{N}_{n m}(\text{vec}(\alpha_i \boldsymbol{M} \boldsymbol{R}_i^\top),  \boldsymbol{R}_i \boldsymbol{\Sigma}_m \boldsymbol{R}_i^\top  \otimes \alpha_i^2 \boldsymbol{\Sigma}_n).
		\label{covN_new}
	\end{align}
\end{definition}
The following notation is also adopted: $||\cdot||$ to indicate the Frobenius norm and $<\cdot,\cdot>$ for the Frobenius inner product \citep{GolubMatrix}.

The main objective of this work is comparing the shapes $\boldsymbol{X}_i$ instead of the form's analysis of the matrices. For that, the parameters of interest are $\boldsymbol{R}_i$ \gramm{and} $\alpha_i$, \gramm{whereas} $\boldsymbol{M}$, $\boldsymbol{\Sigma}_n$, \gramm{and} $\boldsymbol{\Sigma}_m$ are considered as nuisance parameters for each $i = 1, \dots, N$. The estimation of the unknown parameters changes if these nuisance parameters are known. Subsection \ref{krm_GPA} initially presents the estimates under this assumption and then \gramm{provides} the more realistic case of unknown nuisance parameters.

\subsection{Estimation of the Perturbation Model}\label{krm_GPA}

We formalize some results from \cite{Douglas} in the case of known nuisance parameters $\boldsymbol{M}$, $\boldsymbol{\Sigma}_m$, \gramm{and} $\boldsymbol{\Sigma}_n$, with $\boldsymbol{\Sigma}_n$ and $\boldsymbol{\Sigma}_m$ positive definite matrices by the following theorem\gramm{:}

\begin{theorem}[\cite{Douglas}]
	\label{thm:R}
	Consider the perturbation model described in Definition \ref{Procrustes_new}, and the singular value decomposition $\boldsymbol{X}_i^\top \boldsymbol{\Sigma}_n^{-1} \boldsymbol{M} \allowbreak \boldsymbol{\Sigma}_m^{-1}= \boldsymbol{U}_i \boldsymbol{D}_i \boldsymbol{V}_i^\top$. The maximum likelihood estimators equal $\hat{\boldsymbol{R}}_i = \boldsymbol{U}_i \boldsymbol{V}_i^\top$, and $\hat{\alpha_i}_{\boldsymbol{\hat{R}}_i} = ||\boldsymbol{\Sigma}_m^{-1/2} \hat{\boldsymbol{R}}_i^\top \boldsymbol{X}_i^\top \boldsymbol{\Sigma_n}^{-1/2}||^2/tr{(\boldsymbol{D}_i)}$.
\end{theorem}
Now consider $N$ independent observations $\boldsymbol{X}_1, \dots, \boldsymbol{X}_N$. The joint log-likelihood is simply the sum of $N$ log-likelihoods.

In the case of unknown nuisance parameters $\boldsymbol{M}$, $\boldsymbol{\Sigma}_m$\gramm{,} and $\boldsymbol{\Sigma}_n$, the joint likelihood cannot be written as product of separated likelihoods, one for each $\boldsymbol{X}_i$, \gramm{because} each of the unknown parameters is \gramm{a} function of the others. The solution must be found by an iterative algorithm. In particular, the two covariance matrices $\boldsymbol{\Sigma}_m$ and $\boldsymbol{\Sigma}_n$ can be estimated by a two-stage algorithm defined in \cite{dutilleul1999mle}, where $\boldsymbol{\hat{\Sigma}}_n = \{\sum_{i=1}^{N}(\boldsymbol{X}_i - \boldsymbol{\hat{M}}) \boldsymbol{\hat{\Sigma}}_m^{-1} (\boldsymbol{X}_i - \boldsymbol{\hat{M}})^\top\}/Nm$ and 
$\boldsymbol{\hat{\Sigma}}_m = \{\sum_{i=1}^{N}(\boldsymbol{X}_i - \boldsymbol{\hat{M}})^\top \boldsymbol{\hat{\Sigma}}_n^{-1} (\boldsymbol{X}_i - \boldsymbol{\hat{M}})\}/Nn$ are maximum likelihood estimators.

The necessary and sufficient condition for the existence of $\boldsymbol{\hat{\Sigma}}_n$ and $\boldsymbol{\hat{\Sigma}}_m$ is $N \ge\frac{m}{n} + 1$, assuming $\boldsymbol{\Sigma}_m$ and $\boldsymbol{\Sigma}_n$ \gramm{are} positive definite \gramm{matrices}. In real applications, this assumption could be problematic. \angela{For example, in fMRI data analysis, $m$ roughly equals $200,000$, and n approximately equals $200$; therefore, the researcher would have to analyzing at least $1,001$ subjects, which is virtually impossible because fMRI is costly.}

Various solutions can be found in \gramm{the} literature: \cite{Douglas} proposed a regularization for the covariance matrix, \gramm{whereas} \cite{Lele} estimated $\boldsymbol{\Sigma}_n$ using the distribution of $\boldsymbol{X}_i \boldsymbol{X}_i^\top$. In this work, we maintain a general formulation of the estimator for $\boldsymbol{\hat \Sigma}_m = g(\boldsymbol{\hat{\Sigma}}_n, \boldsymbol{\hat{M}}, \boldsymbol{X}_i)$ and $\boldsymbol{\hat \Sigma}_n = g(\boldsymbol{\hat{\Sigma}}_m, \boldsymbol{\hat{M}}, \boldsymbol{X}_i)$, \gramm{because} we aim to find a proper estimate of $\boldsymbol{R}_i$ rather than $\boldsymbol{\Sigma}_n$ and $\boldsymbol{\Sigma}_m$. The shared matrix $\boldsymbol{M}$ is estimated by the element-wise arithmetic mean of $\{\hat{\boldsymbol{X}}_i\}_{i = 1, \dots, N}$, where $\hat{\boldsymbol{X}}_i = \hat{\alpha}_{\hat{\boldsymbol{R}}_i}^{-1} \boldsymbol{X}_i \hat{\boldsymbol{R}}_i$. We then modified the iterative algorithm of \cite{GowerGPA} \gramm{(}i.e., \gramm{g}eneralized Procrustes \gramm{a}nalysis\gramm{)}, to estimate $\boldsymbol{R}_i$. 

\begin{algorithm}
	\begin{algorithmic}[1]
		
	\Require $\boldsymbol{X}_i$, $\mathtt{T}$ , $\mathtt{maxIt}$, $\forall i=1,\dots,N$
	\Ensure $\hat{\boldsymbol{X}}_i$ $\forall i=1,\dots,N$
\State	$\boldsymbol{\hat{M}} = \sum_{i = 1}^{N} \boldsymbol{X}_i /N$, $\hat{\alpha_i} = 1$,  $\boldsymbol{\hat{\Sigma}}_{n\text{; old}} = \hat{\boldsymbol{\Sigma}}_n = \boldsymbol{I}_n$, $\hat{\boldsymbol{\Sigma}}_{m\text{; old}} = \hat{\boldsymbol{\Sigma}}_m = \boldsymbol{I}_m$
\State	$\mathtt{count} = 0$, $\mathtt{dist} = \mathtt{Inf}$
	\While{$\mathtt{dist}$ $>$ $\mathtt{T}$  OR  $\mathtt{count}$ $<$ $\mathtt{maxIt}$}
		\For{$i=1$ to $N$}
		\State	$\boldsymbol{U}_i \boldsymbol{D}_i \boldsymbol{V}_i^\top = \texttt{SVD}(\boldsymbol{X_i}^\top \boldsymbol{\hat{\Sigma}_n}^{-1} \boldsymbol{\hat{M}} \boldsymbol{\hat{\Sigma}_m}^{-1})$ \Comment{\small Singular value decomposition } \label{alg:svd}
		\State	$\boldsymbol{\hat{R}}_i = \boldsymbol{U}_i \boldsymbol{V}_i^\top$
		\State	$\boldsymbol{\hat{X}}_i = \boldsymbol{X}_i \boldsymbol{\hat{R}}_i$ 
		\State	$\hat{\alpha_i}_{\boldsymbol{\hat{R}}_i} = ||\boldsymbol{\hat{\Sigma}}_m^{-1/2} \hat{\boldsymbol{R}}_i^\top \boldsymbol{X}_i^\top \boldsymbol{\hat{\Sigma}_n}^{-1/2}||^2/tr(\boldsymbol{D}_i)$ 
		\State	$\boldsymbol{\hat{X}}_i = \hat{\alpha_i}_{\boldsymbol{\hat{R}}_i}^{-1} \hat{\boldsymbol{X}}_i$ \Comment{\small Update $\boldsymbol{X}_i$}
		
	\EndFor
	\State	$\boldsymbol{\hat{M}}_{\text{old}} = \boldsymbol{\hat{M}}$\Comment{\small Save $\boldsymbol{\hat{M}}$}
	\State	$\boldsymbol{\hat{M}}=\sum_{i = 1}^{N}\boldsymbol{\hat{X}}_i / N$\Comment{\small Update $\boldsymbol{\hat{M}}$}
	\State	$\boldsymbol{\hat{\Sigma}}_n = g(\boldsymbol{\hat{\Sigma}}_m, \boldsymbol{\hat{M}}, \boldsymbol{X}_i)$, $\boldsymbol{\hat{\Sigma}}_m = g(\boldsymbol{\hat{\Sigma}}_n, \boldsymbol{\hat{M}}, \boldsymbol{X}_i) $
		\While{$||\boldsymbol{\hat{\Sigma}}_n - \boldsymbol{\hat{\Sigma}}_{n \text{; old}}||^2  > \epsilon_1$  OR  $||\boldsymbol{\hat{\Sigma}}_m - \boldsymbol{\hat{\Sigma}}_{m \text{; old}}||^2 > \epsilon_2$}
			
		\State	$\boldsymbol{\hat{\Sigma}}_{n \text{; old}} = \boldsymbol{\hat{\Sigma}}_n$, $\boldsymbol{\hat{\Sigma}}_{m \text{; old}} = \boldsymbol{\hat{\Sigma}}_m$
		\State	$\boldsymbol{\hat{\Sigma}}_n = g(\boldsymbol{\hat{\Sigma}}_m, \boldsymbol{\hat{M}}, \boldsymbol{X}_i)$, $\boldsymbol{\hat{\Sigma}}_m = g(\boldsymbol{\hat{\Sigma}}_n, \boldsymbol{\hat{M}}, \boldsymbol{X}_i) $
		\EndWhile
		
	\State	$\mathtt{dist} = ||\boldsymbol{\hat{M}} - \boldsymbol{\hat{M}}_{\text{old}}||^2$, $\mathtt{count} = \mathtt{count} + 1$.
	\EndWhile
	
	\caption{Generalized Procrustes \gramm{a}nalysis \citep{GowerGPA} with the two-stage algorithm from \cite{dutilleul1999mle}, where $\mathtt{T}$ is the threshold for $||\hat{\boldsymbol{M}} - \hat{\boldsymbol{M}}_{old}||^2$, $\mathtt{maxIt}$ is the maximum number of iterations allowed, and $\epsilon_1$\gramm{and} $\epsilon_2$ denote two infinitesimal positive quantities. 
	}\label{alg:algo}
\end{algorithmic}
\end{algorithm}

\cite{groisser2005convergence} proved the convergence of the \gramm{g}eneralized Procrustes \gramm{a}nalysis algorithm. However, it leads to non-identifiable estimators of $\boldsymbol{R}_i$, $i = 1, \dots, N$\gramm{,} proved by the following lemma\gramm{:}

\begin{lemma}\label{thm:u2}
	Let $\{\boldsymbol{\hat{R}}_i\}_{i=1,\dots,N}$ be the maximum likelihood solutions for $\{\boldsymbol{R}_i\}_{i=1,\dots,N}$ with $\boldsymbol{M}$, $\boldsymbol{\Sigma}_n$, and $\boldsymbol{\Sigma}_m$ \gramm{as} unknown parameters. If $\boldsymbol{Z} \in \mathcal{O}(m)$, then $\{\boldsymbol{\hat{R}}_i \boldsymbol{Z}\}_{i = 1, \dots, N}$ are still valid maximum likelihood solutions for $\{\boldsymbol{\hat{R}}_i\}_{i = 1, \dots, N}$. 
\end{lemma}

\begin{proof}
	Consider $\boldsymbol{Z} \in \mathcal{O}(m)$, and
	\begin{equation*}
	\dfrac{1}{\alpha_i} \boldsymbol{X}_i \boldsymbol{R}_i \boldsymbol{Z} - \boldsymbol{M}\boldsymbol{Z} = \boldsymbol{E}_i \boldsymbol{Z} \sim \mathcal{MN}(0, \boldsymbol{\Sigma}_n, \boldsymbol{Z}^\top \boldsymbol{\Sigma}_m \boldsymbol{Z}).
	\end{equation*}
	Consider the proof of Theorem \ref{thm:R} \angela{placed in the appendix}, we have
	\begin{align}
	\max_{\boldsymbol{R}_i \in \mathcal{O}(m)} \sum_{i = 1}^{N} <\boldsymbol{R}_i, \boldsymbol{X}_i^\top \boldsymbol{\Sigma}_n^{-1} \boldsymbol{M} \boldsymbol{\Sigma}_m^{-1}> \label{min} = \max_{\boldsymbol{R}_i \in \mathcal{O}(m)} \sum_{i =1}^{N} tr(\boldsymbol{Z}^\top \boldsymbol{R}_i^\top \boldsymbol{X}_i^\top \boldsymbol{\Sigma}_n^{-1} \boldsymbol{M} \boldsymbol{Z} \boldsymbol{Z}^\top \boldsymbol{\Sigma}_m^{-1} \boldsymbol{Z}). 
	\end{align}
	Since $\boldsymbol{Z}^\top \boldsymbol{Z} = \boldsymbol{Z} \boldsymbol{Z}^\top = \boldsymbol{I}_m$, the solutions $\{\hat{\boldsymbol{R}}_i \boldsymbol{Z}\}_{i = 1, \dots, N}$ are still valid solutions for the maximization \eqref{min}.
\end{proof}

To sum up, in the more realistic case, when we must estimate the nuisance parameters, the Procrustes solutions are infinite in general, \gramm{in} both in the high-dimensional case \gramm{and} the low-dimensional by Lemma \ref{thm:u2}. We emphasize here that, to resolve the non-identifiability of $\boldsymbol{R}_i$, the proposed ProMises model imposes a prior distribution for the parameter $\boldsymbol{R}_i$. 

\section{ProMises Model}\label{procrustes-solution-with-regularization}
\subsection{Background}
In the previous section, we justified how the perturbation model could be problematic \gramm{because} Lemma \ref{thm:u2} proves the non-identifiability of the parameter $\boldsymbol{R}_i$ in the realistic case of unknown nuisance parameters. This scenario returns to be critical in several applications, where the $m$ columns of the matrix $\boldsymbol{X}_i$ do not express the three-dimensional spatial coordinates, such as in the fMRI data framework illustrated in Section \ref{application}.
In the high-dimensional case, the final orientations of the aligned data can be relevant for interpretation purposes. 

For that, we propose its Bayesian approach: the ProMises model. We stress here that a proper prior distribution leads to a closed-form and interpretable, unique point estimate of $\boldsymbol{R}_i$. The specification of the prior parameters is essential, especially in our high-dimensional context.

\subsection{Interpretation of the Prior Parameters}\label{lmp}

\gramm{Because} $\boldsymbol{R}_i \in \mathcal{O}(m)$, a proper prior distribution must take values in the Stiefel manifold $V_{m}({\rm I\!R}^m)$. The matrix von Mises-Fisher distribution is a non-uniform distribution on $V_{m}({\rm I\!R}^m)$, which describes a rigid configuration of $m$ distinct directions with fixed angles. It was proposed by \cite{Downs} and investigated by many authors \gramm{(}e.g., \cite{Jupp, Chikuse2}\gramm{)}. 

We report below the formal definition of the von Mises-Fisher distribution.

\begin{definition}[\cite{Downs}]\label{fisher_def}
	The von Mises-Fisher distribution for $\boldsymbol{R}_i \in \mathcal{O}(m)$ is    
	\begin{align}
		f(\boldsymbol{R}_i) = C(\boldsymbol{F}, k ) \exp\big\{tr(k \boldsymbol{F}^\top \boldsymbol{R}_i)\big\}
		\label{fisher},
	\end{align}
	where \(C(\boldsymbol{F}, k)\) is a normalizing constant, $\boldsymbol{F} \in {\rm I\!R}^{m \times m}$ is the location matrix parameter, and $k \in {\rm I\!R}^{+}$ is the concentration parameter. 
\end{definition}

The parameter $k$ defined in \eqref{fisher} balances the amount of concentration of the distribution around $\boldsymbol{F}$. If $k \rightarrow 0$, the prior distribution is near a uniform distribution \gramm{(} i.e., unconstrained\gramm{)}. If $k \rightarrow +\infty$, the prior distribution tends toward a Dirac distribution \gramm{(}i.e., maximum constraint\gramm{)}.

A proper specification of the prior distribution leads to improved estimation of $\boldsymbol{R}_i$. Therefore, the core of the ProMises model is the specification of $\boldsymbol{F}$ defined in \eqref{fisher}. Consider the polar decomposition and singular value decomposition of $\boldsymbol{F}= \boldsymbol{P} \boldsymbol{K} = \boldsymbol{L} \boldsymbol{\Sigma}\boldsymbol{B}^\top = \boldsymbol{L}  \boldsymbol{B}^\top \boldsymbol{B} \boldsymbol{\Sigma} \boldsymbol{B}^\top$, where $\boldsymbol{P}, \boldsymbol{L}, \boldsymbol{B} \in \mathcal{O}(m)$, \gramm{and} $\boldsymbol{K} \in {\rm I\!R}^{m \times m}$ \gramm{are} symmetric positive semi-definite matrices, and $\boldsymbol{\Sigma} \in {\rm I\!R}^{m \times m}$ diagonal matrix with non-negative real numbers on the diagonal. The mode of the density defined in \eqref{fisher} equals $\boldsymbol{P}$ \citep{Jupp}, so the most plausible rotation matrix depends on the orientation characteristic of $\boldsymbol{F}$. Merging the two decompositions, $\boldsymbol{P} = \boldsymbol{L} \boldsymbol{B}^\top$ describes the orientation part of $\boldsymbol{F}$, and $\boldsymbol{K} = \boldsymbol{B} \boldsymbol{\Sigma} \boldsymbol{B}^\top$ defines the concentration part. The mode is specified by the product \gramm{of} the left and right singular vectors of $\boldsymbol{F}$. These decompositions are useful to understand when the density $\eqref{fisher}$ is uni-modal. If $\boldsymbol{F}$ has full rank,  $\boldsymbol{\Sigma}$ does too, the \gramm{p}olar \gramm{d}ecomposition is unique, \gramm{and thus} the mode of the density \gramm{(}i.e., $\boldsymbol{P}$ is the global maximum\gramm{)}. Let $\boldsymbol{F}$ be a full rank matrix, then the maximum equals $\max_{\boldsymbol{R}_i \in \mathcal{O}(m)} \,\, tr(\boldsymbol{F} \boldsymbol{R}_i^\top) = tr\{\boldsymbol{L}  \boldsymbol{B}^\top \boldsymbol{B} \boldsymbol{\Sigma} \boldsymbol{B}^\top (\boldsymbol{L} \boldsymbol{B}^\top)^\top\} = tr(\boldsymbol{\Sigma})$.

To sum up, the prior specification allows us to include a priori information about the optimal orientation in the perturbation model. We anticipate here the result of Lemma \ref{thm:un}: \gramm{I}f $\boldsymbol{F}$ is defined as a full rank matrix, the maximum a posteriori solution $\boldsymbol{\hat{R}^{'}}_i$ will be unique with the orientation structure of $\boldsymbol{F}$.

Two simple examples of $\boldsymbol{F}$ are delineated below.

\begin{example}\label{example}
	The most simple definition of $\boldsymbol{F}$ is $\boldsymbol{I}_m$ \citep{Lee}. The eigenvalues are all $1$, and $\boldsymbol{L}$ and $\boldsymbol{B}$ are equal to $\boldsymbol{e}_1,\dots, \boldsymbol{e}_m$, where $\boldsymbol{e}_i$ is the standard basis forming an orthonormal basis of ${\rm I\!R}^{m}$. The prior distribution shrinks the possible solutions for $\boldsymbol{R}_i$ toward orthogonal matrices that consider only the combination of variables with the same location.
	
	\angela{Alternatively, considering the fMRI scenario, the hyperparameter $\boldsymbol{F}$ can be defined as an Euclidean similarity matrix using the $3D$ anatomical coordinates $x$, $y$ and $z$ of each voxel:
	\begin{align*}
	\boldsymbol{F} = \Big[\exp\Big\{- \sqrt{(x_i - x_j)^2 + (y_i - y_j)^2 + (z_i - z_j)^2}\Big\}\Big],
	\end{align*}
	where $i,j = 1,\dots v$. In this way, $\boldsymbol{F}$ is a symmetric matrix with ones in the diagonal, which means that voxels with the same spatial location are combined with weights equalling $1$, and the weights decrease as the voxels to be combined are more spatially distant.}
\end{example}

\begin{example}\label{examplePlant}
	Consider $N$ matrices, one for each plant, describing the three-dimensional spatial trajectories of a climbing plant, Pisum sativum, having wooden support as a stimulus \citep{Castiello} \gramm{across time}. The spatiotemporal trajectories of the plants are analyzed until they come to grasp the stick. The aim is to functionally align the three time series, one for each coordinate ($x$, $y$, $z$), then we have $m=3$. In this case, we could suppose that the rotation along the $z$ axis is not of interest \gramm{because} the functional misalignment between plants can be along the $x$\gramm{-axis} and $y$\gramm{-axis} \gramm{with} the $z$\gramm{-}axis \gramm{being} the one that reflects the growth of the plants and the $x$\gramm{-axis} and $y$\gramm{-axis} describ\gramm{ing} the elliptical movement \gramm{(}circumnutation\gramm{)} of the plants \citep{Castiello}. Therefore, the $\boldsymbol{F}$ location matrix parameter can be described as \gramm{follows}:
	\begin{equation}\label{plant_F}
		\boldsymbol{F}=\begin{bmatrix}
			0.5 & 0.5 & 0 \\
			0.5 & 0.5 & 0 \\
			0 & 0 & 1 \\
		\end{bmatrix}.
	\end{equation}
	The axes $x$ and $y$ have the same probability of entering in the calculation of the first two dimensions of the final common space, \gramm{whereas} the $z$ axis is not considered.
	
	We use the kinematic plant data from \cite{Castiello}, consisting of \gramm{five} matrices/plants. Figure \ref{fig:plant} shows the elliptical movement expressed by the \gramm{axes} $x$ and $y$, in the case of unaligned and aligned plant trajectories. The rotation transformations are estimated by the ProMises model with $\boldsymbol{F}$ expressed as \eqref{plant_F}. We do not go into detail about the meaning of the results \gramm{because} we have introduced this example to explain the usefulness of $\boldsymbol{F}$. However, we can note how the ProMises model aligns the final coordinates of the tendrils \gramm{(}i.e., when the plant touches the wooden support\gramm{)}.

\begin{figure}[!htb]
	\begin{minipage}{.5\textwidth}
		\centering
		\includegraphics[width=\textwidth]{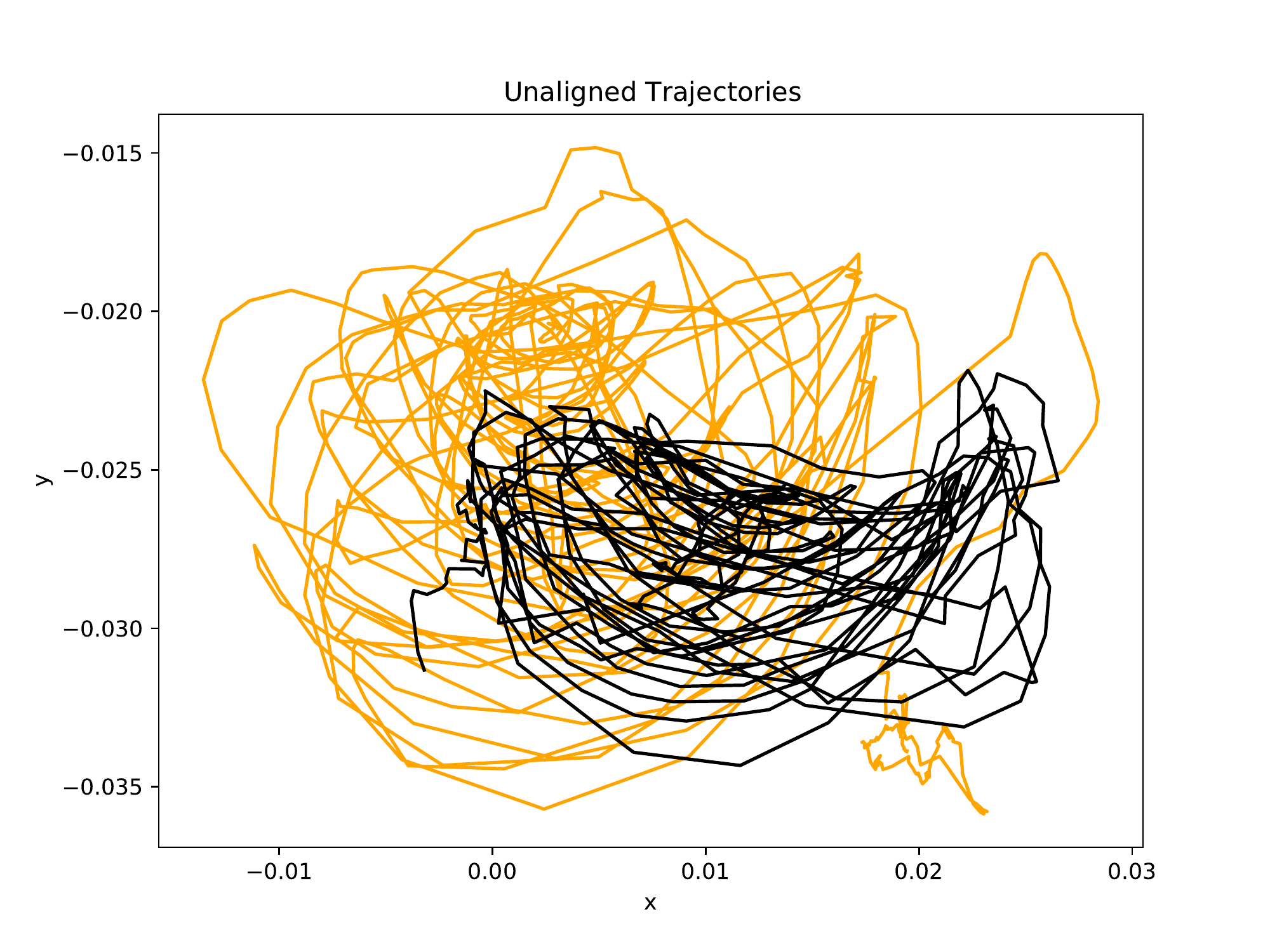}
	\end{minipage}%
	\begin{minipage}{0.5\textwidth}
		\centering
		\includegraphics[width=\textwidth]{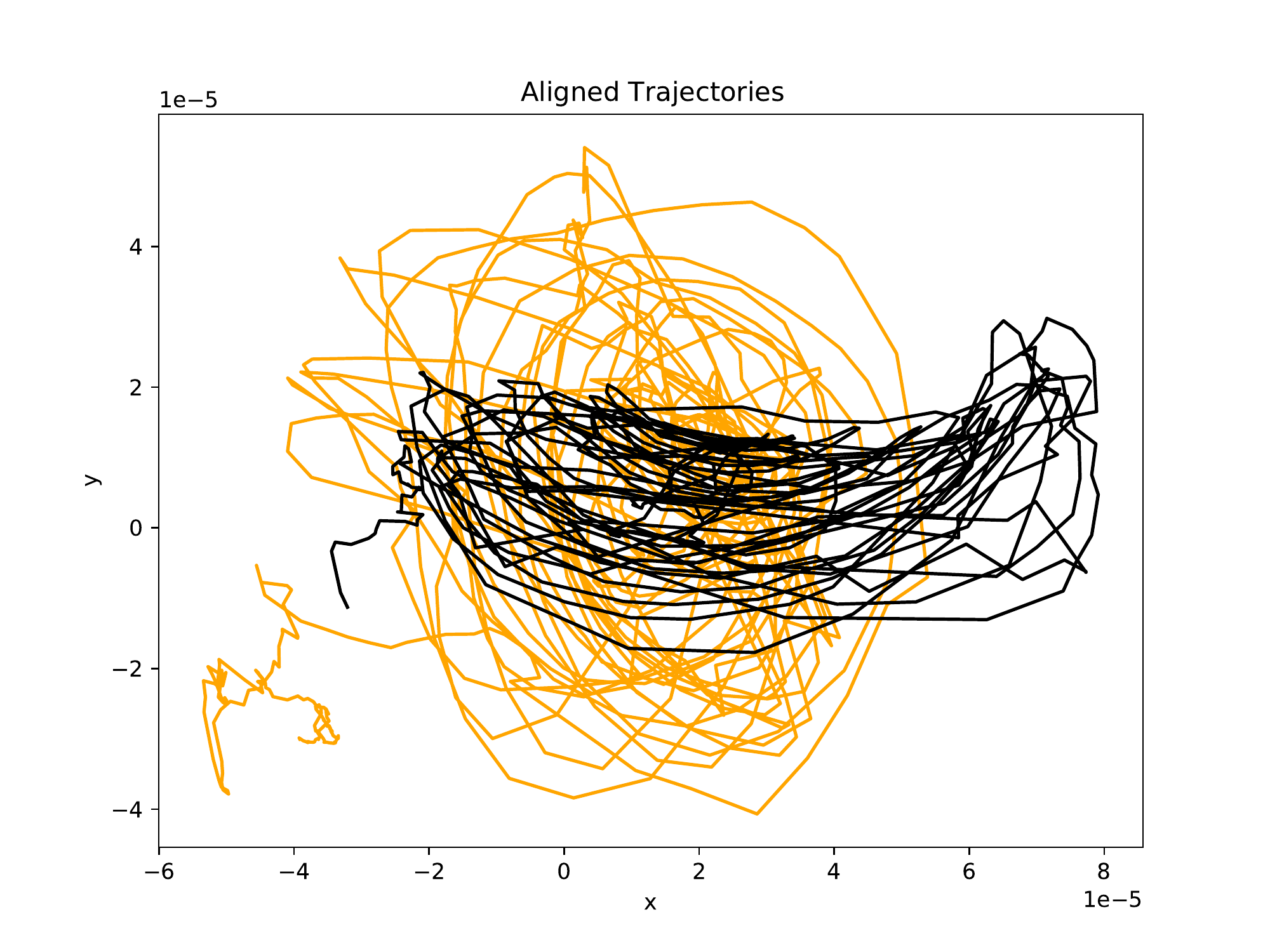}
		
	\end{minipage}
	\caption{Left panel: Unaligned spatial trajectories of the tendrils of two plants. Right panel: Aligned spatial trajectories of the tendrils of two plants.}
	\label{fig:plant}
\end{figure}
\end{example}

\subsection{Von Mises-Fisher Conjugate Prior}\label{maximum-posterior-estimation}
The von Mises-Fisher distribution \eqref{fisher} was proved by \cite{Khatri} to be a member of the standard exponential family \citep{Barndorff}. \cite{MardiaBayesian1} mentioned that the von Mises-Fisher distribution is a conjugate prior for the matrix normal distribution, \gramm{which} we formally prove in the following lemma under the perturbation model's assumptions.
\begin{lemma}\label{lemma3}
	Consider the perturbation model of Definition \ref{Procrustes_new}, with $\boldsymbol{R}_i$ distributed \angela{according} to \eqref{fisher}, then the posterior distribution $f(\boldsymbol{R}_i| k, \boldsymbol{F}, \boldsymbol{X}_i)$ is \gramm{a} conjugate distribution to the von Mises-Fisher prior distribution with location posterior parameter equal\gramm{ling the following:}
	\begin{align}\label{eq:eq1111}
		\boldsymbol{F}^\star=\boldsymbol{X_i}^\top \boldsymbol{\Sigma}_n^{-1} \boldsymbol{M} \boldsymbol{\Sigma}_m^{-1} + k \boldsymbol{F}.
	\end{align}
\end{lemma}

The posterior location parameter is the sum of $\boldsymbol{X}_i^\top \boldsymbol{\Sigma}_n^{-1} \boldsymbol{M} \boldsymbol{\Sigma}_m^{-1}$ and the prior location parameter $\boldsymbol{F}$ multiplied by $k$.
Consider the singular value decomposition of $\boldsymbol{X}_i^\top \boldsymbol{\Sigma}_n^{-1} \boldsymbol{M} \boldsymbol{\Sigma}_m^{-1}$:
\begin{align}
	\boldsymbol{X}_i^\top \boldsymbol{\Sigma}_n^{-1} \boldsymbol{M} \boldsymbol{\Sigma}_m^{-1} = \boldsymbol{U}_i \boldsymbol{D}_i \boldsymbol{V}_i^\top = \boldsymbol{U}_i \boldsymbol{V}_i^\top \boldsymbol{V}_i \boldsymbol{D}_i \boldsymbol{V}_i^\top.
	\label{eq:eq10}
\end{align}
The right part of \eqref{eq:eq10} $\boldsymbol{V}_i \boldsymbol{D}_i \boldsymbol{V}_i^\top$ is the elliptical part of $\boldsymbol{X}_i^\top \boldsymbol{\Sigma}_n^{-1} \boldsymbol{M} \boldsymbol{\Sigma}_m^{-1}$, which is a measure of variation relative to the decomposition $\boldsymbol{U}_i \boldsymbol{V}_i^\top$ \gramm{(}i.e., the maximum likelihood estimator of $\boldsymbol{R}_i$\gramm{)}. 
Focus on the right part of \eqref{eq:eq1111}, $\boldsymbol{F}=\boldsymbol{P} \boldsymbol{K}$\gramm{, which is} the \gramm{p}olar \gramm{d}ecomposition of $\boldsymbol{F}$, where $\boldsymbol{P}$ is the mode of the von Mises-Fisher distribution, and $\boldsymbol{K}$ \gramm{is} its measure of variation. Therefore, $\boldsymbol{F}^\star$ is expressed as a combination \gramm{of} the maximum likelihood estimate $\boldsymbol{\hat{R}}_i = \boldsymbol{U}_i \boldsymbol{V}_i^\top$ and the prior mode $\boldsymbol{P}$, multiplied by corresponding measures of variation.

Thanks to the conjugacy, the estimation process remains simple\gramm{;} with a small modification we keep the previous algorithm without increasing the computational burden.

\subsection{Estimation of the ProMises Model}\label{orthogonal-matrix-estimation}
This section delineates the estimation process for $\boldsymbol{R}_i$ using the ProMises method. First of all, $f(\boldsymbol{X_i} | \alpha_i, \boldsymbol{R_i})$ depends only on the product $\alpha_i \boldsymbol{R_i}$, we thus refer to the distribution $f(\boldsymbol{X_i} | \alpha_i \boldsymbol{R_i})$ instead of $f(\boldsymbol{X_i} | \alpha_i, \boldsymbol{R_i})$ defined in \eqref{covN_new}. The following density is then considered \gramm{as prior distribution for the product $\alpha_i \boldsymbol{R_i}$}:
\begin{align}
	f(\alpha_i \boldsymbol{R_i}) \sim \exp\Big\{\dfrac{k}{\alpha_i}tr(\boldsymbol{F}^\top \boldsymbol{R}_i)\Big\}\alpha_i^{-1}.
	\label{prior}
\end{align}

The following theorem delineates the estimation of $\boldsymbol{R}_i$ with known nuisance parameters\gramm{:}
\begin{theorem}\label{thm3}
	The ProMises model is defined as the perturbation model \eqref{eq:eq15_new} imposing the prior distribution \eqref{prior} for $\alpha_i \boldsymbol{R}_i$. Let the singular value decomposition of $\boldsymbol{X}_i^\top \boldsymbol{\Sigma}_n^{-1} \boldsymbol{M} \boldsymbol{\Sigma}_m^{-1} + k \boldsymbol{F}$ be $\boldsymbol{U}_i \boldsymbol{D}_i \boldsymbol{V}_i^\top$. Then\gramm{,} the maximum a posteriori estimators equal $\boldsymbol{\hat{R}^{'}}_i = \boldsymbol{U}_i \boldsymbol{V}_i^\top$ and \\ $\hat{\alpha_i}_{\boldsymbol{\hat{R}^{'}}_i}^{'}=||\boldsymbol{\Sigma}_m^{-1/2} \boldsymbol{\hat{R}^{'\top}}_i \boldsymbol{X}_i^\top \boldsymbol{\Sigma}_n^{-1/2}||^2/tr(\boldsymbol{D}_i)$.
\end{theorem}

The prior information about $\boldsymbol{R}_i$'s structure \gramm{is directly entered} in the singular value decomposition step\gramm{;} the maximum a posteriori estimator turns out to be a slight modification of the solution given in Theorem \ref{thm:R}. We decompose $\boldsymbol{X}_i^\top \boldsymbol{\Sigma}_n^{-1} \boldsymbol{M} \boldsymbol{\Sigma}_m^{-1} + k \boldsymbol{F}$ instead of $\boldsymbol{X}_i^\top \boldsymbol{\Sigma}_n^{-1} \boldsymbol{M} \boldsymbol{\Sigma}_m^{-1}$.

Let $\{\boldsymbol{X}_{i}  \in {\rm I\!R}^{n \times m} \}_{i = 1,\dots,N}$ be a set of independent matrices. Then the joint posterior distribution is simply the product of the single posterior distribution. 

If $\boldsymbol{M}$, $\boldsymbol{\Sigma}_n$, \gramm{and} $\boldsymbol{\Sigma}_m$ are unknown, the maximization problem has no closed-form solution, \gramm{like} in Section \ref{krm_GPA}. \gramm{Because} we proved that the prior specification modifies only the singular value decomposition step of Theorem \ref{thm:R}, we then modify the \gramm{L}ine \ref{alg:svd} of Algorithm \ref{alg:algo}\gramm{, as follows}:


\begin{algorithm}

	\caption{ProMises model.
	}\label{alg:algo_promises}
 
	Use Algorithm \ref{alg:algo}, only modify:
	\begin{algorithmic}[1]
	\setalglineno{5}
	\State $\boldsymbol{U}_i \boldsymbol{D}_i \boldsymbol{V}_i^\top = \texttt{SVD}(\boldsymbol{X_i}^\top \boldsymbol{\hat{\Sigma}_n}^{-1} \boldsymbol{\hat{M}} \boldsymbol{\hat{\Sigma}_m}^{-1} + k \boldsymbol{F})$\Comment{\small Singular value decomposition }
	\end{algorithmic}
\end{algorithm}

\subsection{On the Choice of the Parameter of the von Mises-Fisher Distribution}\label{choice}
We choose the von Mises-Fisher distribution as prior distribution for $\boldsymbol{R}_i$ \gramm{for its} useful and practical properties. First, as shown, it is a conjugate prior distribution, leading to a direct calculation and interpretation of $\boldsymbol{\hat{R}^{'}}_i$. Second, it expresses the orthogonality constraint imposed by the Procrustes problem. Finally, the definition of $\boldsymbol{F}$ does not require strong assumptions \citep{Downs}; nevertheless, if we specify it as a full-rank matrix, we guarantee \gramm{solution's uniqueness}. This permits formulat\gramm{ion of} the \gramm{below} lemma.

\begin{lemma}\label{thm:un}
	If $\boldsymbol{F}$ has full rank, the maximum a posteriori estimates for $\boldsymbol{R}_i$ given by Theorem \ref{thm3} are unique.
\end{lemma}

\begin{proof}
	Consider the proof of Lemma \ref{thm:u2}. Multiplying by $\boldsymbol{Z}$ leads to the following maximization:
	\begin{align*}
	\sum_{i=1}^{N} tr(\boldsymbol{Z}^\top \boldsymbol{R}_i^\top (\boldsymbol{X}_i^\top \boldsymbol{\Sigma}_n^{-1} \boldsymbol{M} \boldsymbol{Z} \boldsymbol{Z}^\top \boldsymbol{\Sigma}_m^{-1} \boldsymbol{Z} + k \boldsymbol{F})) \ne \sum_{i=1}^{N} tr( \boldsymbol{R}_i^\top (\boldsymbol{X}_i^\top \boldsymbol{\Sigma}_n^{-1} \boldsymbol{M}  \boldsymbol{\Sigma}_m^{-1} + k \boldsymbol{F}))
	\end{align*} 
	since the cyclic permutation invariance property of the trace does not work as Lemma \ref{thm:u2} having the additional term $k tr( \boldsymbol{F} \boldsymbol{R}_i)$. 
	
	In addition, recalling Lemma \ref{thm:u1}, the solution for $\boldsymbol{R}_i$ is unique if and only if $\boldsymbol{X}_i^\top \boldsymbol{\Sigma}_n^{-1} \boldsymbol{M} \boldsymbol{\Sigma}_m^{-1}+ k \boldsymbol{F}$ has full rank. If $\boldsymbol{F}$ is defined with full rank, so $\boldsymbol{\tilde{X}}_i^\top \boldsymbol{\tilde{M}} + k \boldsymbol{F}$, and the solution for $R_i$ is unique. Furthermore, recalling \cite{Jupp}, the mode of the von Mises-Fisher is the orientation part of location matrix parameter. \gramm{Because} the polar decomposition of $\boldsymbol{F}^\star$ is unique, the maximum a posteriori estimate is unique.
\end{proof}

To sum up, the ProMises model enables resolving the non-identifiability of $\boldsymbol{R_i}$ that characterizes the perturbation model. The prior information inserted in the model permits \gramm{guidance of} the estimation process, computing a unique and interpretable data orthogonal transformation. Finally, all these properties are reached without complicating the estimation process of the perturbation model, we only modify the singular value decomposition step of Algorithm \ref{alg:algo}.

\section{Efficient ProMises Model}\label{light}
The framework depicted above can be applied both in low and high-dimensional settings. However, the extension to the high-dimensional case does not come for free if the perturbation model is used. When $n<m$, rank equals $n$, the identifiability of the solution is lost even when the nuisance parameters are known. The lemma below \gramm{formally} states\gramm{:}

\begin{lemma}\label{thm:u1}
	Consider $\boldsymbol{X}_i \in {\rm I\!R}^{n \times m}$, if $n < m$, then the maximum likelihood estimate for $\boldsymbol{R}_i$ defined in Theorem \ref{thm:R} is not unique.
\end{lemma}

\angela{Although} the ProMises model provides unique solutions even in high-dimensional frameworks, a second issue remains prominent: the computational load. At each step, the presented algorithms perform $N$ singular values decompositions of $m\times m$ matrices, which have a polynomial-time complexity $O(m^3)$. When $m$ becomes large, as in fMRI data where $m$ is a few hundred thousands, the computation runtime, \gramm{and} the required storing memory, \gramm{become} inadmissible.

This section proposes the Efficient ProMises model, which resolves the two \gramm{above} points. The method is efficient in terms of space and time complexity and fixes the non-identifiability of $\boldsymbol{R}_i$. The algorithm allows a faster and more accessible shape analysis without loss of information in the case of $n \ll m$. It essentially merges the thin singular value decomposition \citep{bai2000templates} with the Procrustes problem.

In practice, the Efficient ProMises approach projects  the matrices $\boldsymbol{X}_i$ into a\gramm{n} $n$\gramm{-}lower-dimensional space \gramm{using} a specific semi-orthogonal transformation \citep{Abadir, Grob} $\boldsymbol{Q}_i$, with dimensions $m \times n$, which preserve all the data's information. It aligns, then, the reduced $n\times n$ matrices $\{\boldsymbol{X}_i \boldsymbol{Q}_i \in {\rm I\!R}^{n \times n }\}_{i = 1, \dots, N}$ by the perturbation or ProMises model. Finally, it projects the aligned matrices \gramm{back} to the original $n\times m$-size matrices $\{\boldsymbol{X}_i \in {\rm I\!R}^{n \times m }\}_{i = 1, \dots, N}$ using the transpose of $\{\boldsymbol{Q}_i\}_{i = 1, \dots, N}$.

The following theorem proves that the maximum \angela{defined in Equation \eqref{min}} using $\{\boldsymbol{X}_i \boldsymbol{Q}_i \in {\rm I\!R}^{n \times n }\}_{i = 1, \dots, N}$ equals the original \gramm{maximum} \gramm{because} the Procrustes problem analyzes the first $n \times n$ dimensions of $\boldsymbol{R}_i$. The maximum remains the same if we multiply $\{\boldsymbol{X}_i \in {\rm I\!R}^{n \times m }\}_{i = 1, \dots, N} $ by $\boldsymbol{Q}_i$.

\begin{theorem}\label{thm1}
	Consider the perturbation model in Definition \ref{Procrustes_new} with $\boldsymbol{\Sigma}_m = \sigma^2 \boldsymbol{I}_m$ and the thin singular value decompositions of $\boldsymbol{X}_i = \boldsymbol{L}_i \boldsymbol{S}_i \boldsymbol{Q}_i^\top$ for each $i = 1, \dots, N$, where $\boldsymbol{Q}_i$ has dimensions $m \times n$. The following holds
	\begin{align*}
		\max_{\boldsymbol{R}_i \in \mathcal{O}(m)} tr(\boldsymbol{R}_i^\top \boldsymbol{X}_i^\top \boldsymbol{\Sigma}_n^{-1} \boldsymbol{X}_j \boldsymbol{\Sigma}_m^{-1}) = \max_{\boldsymbol{R}_i^{\star} \in \mathcal{O}(n)} tr(\boldsymbol{R}_i^{ \star \top} \boldsymbol{Q}_i^\top \boldsymbol{X}_i^\top \boldsymbol{\Sigma}_n^{-1} \boldsymbol{X}_j \boldsymbol{\Sigma}_m^{-1} \boldsymbol{Q}_j^\top).
	\end{align*}
	
\end{theorem}

\gramm{Additionally}, the condition for the existence of $\hat{\boldsymbol{\Sigma}}_n$ by \cite{dutilleul1999mle} is satisfied \gramm{because} $\boldsymbol{X}_i$ now \gramm{has} dimensions $n \times n$, and the functional alignment needs at least $N = 2$ observations.

So, Theorem \ref{thm1} is used to define an Efficient version of the ProMises model.

\begin{lemma}\label{light_vmp}
	Consider the assumptions of Theorem \ref{thm1}, then 
	\begin{align*}
		\max_{\boldsymbol{R}_i \in \mathcal{O}(m)} tr(\boldsymbol{R}_i^\top \boldsymbol{X}_i^\top \boldsymbol{\Sigma_n}^{-1} \boldsymbol{X}_j \boldsymbol{\Sigma}_m^{-1} + k \boldsymbol{F}) = \max_{\boldsymbol{R}_i^{\star} \in \mathcal{O}(n)} tr\{\boldsymbol{R}_i^{\star \top} (\boldsymbol{Q}_i^\top \boldsymbol{X}_i^\top \boldsymbol{\Sigma}_n^{-1} \boldsymbol{X}_j \boldsymbol{\Sigma}_m^{-1} \boldsymbol{Q}_j^\top + k \boldsymbol{F}^\star)\},
	\end{align*}
	where $\boldsymbol{F} \in {\rm I\!R}^{m \times m}$ and $\boldsymbol{F}^\star = \boldsymbol{Q}^\top_i \boldsymbol{F} \boldsymbol{Q}_j \in {\rm I\!R}^{n \times n}$.
\end{lemma}

Proofs of Theorem \ref{thm1} and Lemma \ref{light_vmp} are shown in the appendix, \gramm{whereas} here, we make some further considerations about the proposed method.

\gramm{At first glance,} \gramm{t}he assumption $\boldsymbol{\Sigma}_m = \sigma^2 \boldsymbol{I}_m$ may dilute \gramm{the resul's impact}. However, this assumption does not imply that the data are column-wise independent \gramm{because} this dependence is modeled by $\boldsymbol{R}_i$. \gramm{Additionally}, the joint and accurate estimate of $\boldsymbol{\Sigma}_m$, $\boldsymbol{\Sigma}_n$ and $\boldsymbol{R}_i$ requires a large number of observations. So, it is common in real applications to set $\boldsymbol{\Sigma}_m$ and $\boldsymbol{\Sigma}_n$ to be proportional to the identities in Procrustes-like problems \citep{haxby2020hyperalignment}. When the model is high-dimensional, this problem becomes even more pronounced \gramm{because of} the huge number of parameters to be estimated.

The Efficient ProMises approach reaches the same maximum while working in the reduced space
of the first $n$ eigenvectors, which contains all the information, instead of the full \gramm{data} set. Therefore, the original problem estimates orthogonal matrices of size $m\times m$:
$\boldsymbol{R}_i \in \mathcal{O}(m)$, \gramm{whereas} the Efficient solution provides a set of orthogonal matrices of size $n\times n$: $\boldsymbol{R}_i^{\star} \in \mathcal{O}(n)$.
Even when the solution is projected back into the $m\times m$ space through $\boldsymbol{Q}_i\boldsymbol{R}_i^{\star}\boldsymbol{Q}_i^\top$, the rank remains $n$, \gramm{whereas} the matrices of the original solutions have rank $m$. 
This should clarify that the Efficient approach reaches the same fit to the data under a different set of constraints\gramm{that is} $n\times n$ orthogonal matrices instead of $m\times m$ \gramm{matrices;} hence\gramm{,} the solutions of the two algorithms will not be identical. 

Then\gramm{,} we add on Algorithm \ref{alg:algo} the lines \gramm{used} to reduce the dimensions of $\boldsymbol{X}_i$, and we modify Line \ref{alg:svd} to insert the prior information: 


\begin{algorithm}[H]

	\caption{Efficient ProMises model.
	}\label{algo2}
	
	Use Algorithm \ref{alg:algo}, only add these three lines at the beggining:
	\begin{algorithmic}[1]
		\For{$i=1$ to $N$}
		\State $\boldsymbol{L}_i \boldsymbol{S}_i \boldsymbol{Q}_i^\top = \texttt{SVD}(\boldsymbol{X_i})$ \Comment{\small Thin singular value decomposition}	
		\State $\boldsymbol{X}_i  := \boldsymbol{X_i} \boldsymbol{Q}_i$
		\EndFor
\end{algorithmic}	
		and modify:
\begin{algorithmic}[1]
		\setalglineno{5}
		\State $\boldsymbol{U}_i \boldsymbol{D}_i \boldsymbol{V}_i^\top = \texttt{SVD}(\boldsymbol{X_i}^\top \boldsymbol{\hat{\Sigma}_n}^{-1} \boldsymbol{\hat{M}} \boldsymbol{\hat{\Sigma}_m}^{-1} + k \boldsymbol{F})$\Comment{\small Singular value decomposition }
	\end{algorithmic}
\end{algorithm}

The Efficient approach reduces the time complexity from $\mathcal{O}(m^3)$ to $\mathcal{O}(m n^2)$, and the space complexity from $\mathcal{O}(m^2)$ to $\mathcal{O}(m n)$.

\section{Functional Magnetic Resonance Imaging Data Application}\label{application}


\subsection{Motivation}\label{motivation}

The alignment problem is recognized in fMRI multi-subject studies \gramm{because} the brain's anatomical and functional structure\gramm{s} var\gramm{y} across subjects. The most used anatomical alignment\angela{s} \angela{are} the Talairach normalization \citep{Tal} and \angela{the Montréal Neurological Institute (MNI) space normalization} \citep{Jenkinson}, where the brain images are aligned to an anatomical template by affine transformations using a set of major anatomical landmarks. However, this alignment does not explore the between-subjects \gramm{variability in anatomical positions of the functional loci}. The functional brain regions are not consistently placed on the anatomical landmarks defined by the Talairach \angela{and MNI} template\angela{s}. The anatomical alignment is then an approximate inter-subject registration of the functional cortical areas. \cite{Haxby} proved that functional brain anatomy exhibits a regular organization at a fine spatial scale shared across subjects.

Therefore, we can assume that anatomical and functional structures are subject-specific \citep{HaxbyConn,Sabuncu} and \gramm{that} the neural activities in different brains are noisy rotations of a common space \citep{Haxby}. Functional alignments \gramm{(}e.g., Procrustes methods\gramm{)} attempt to rotate the neural activities to maximize similarity across subjects.

Specifically, each subject's brain activation can be represented by a matrix, where the rows represent the stimuli/time points, and the columns represent the voxels. The stimuli are time-synchronized \gramm{among} subjects, so we have correspondence \gramm{among} the matrices' rows. However, the columns are not assumed to be in correspondence \gramm{among} subjects, as explained before. Each time series of brain activation \gramm{(}i.e., \gramm{each of} the matrices' columns\gramm{)} represents the voxels' functional characteristics that the \angela{anatomical} normalization fails to align. We \gramm{aim} to represent the neural responses to stimuli into a common high-dimensional space, rather than in a canonical anatomical space that does not consider the variability of functional topographies loci.

Figure \ref{fig:toy} shows three voxels' neural activities \gramm{(}i.e., $v_1$, $v_2$, and $v_3$\gramm{)} three columns of the data matrix in two subjects recorded across time. The functional pattern of $v_2$ is equal across subjects, \gramm{whereas} $v_1$ and $v_3$ are swapped. A rotation matrix can resolve this misalignment\gramm{, with}  the swap \gramm{being} a particular case of the rotation matrix. For further details about the motivation in using functional Procrustes-based alignment in fMRI studies, see \cite{haxby2020hyperalignment}.

\begin{figure}[!ht]
	\centering	
	\includegraphics[width=.7\linewidth]{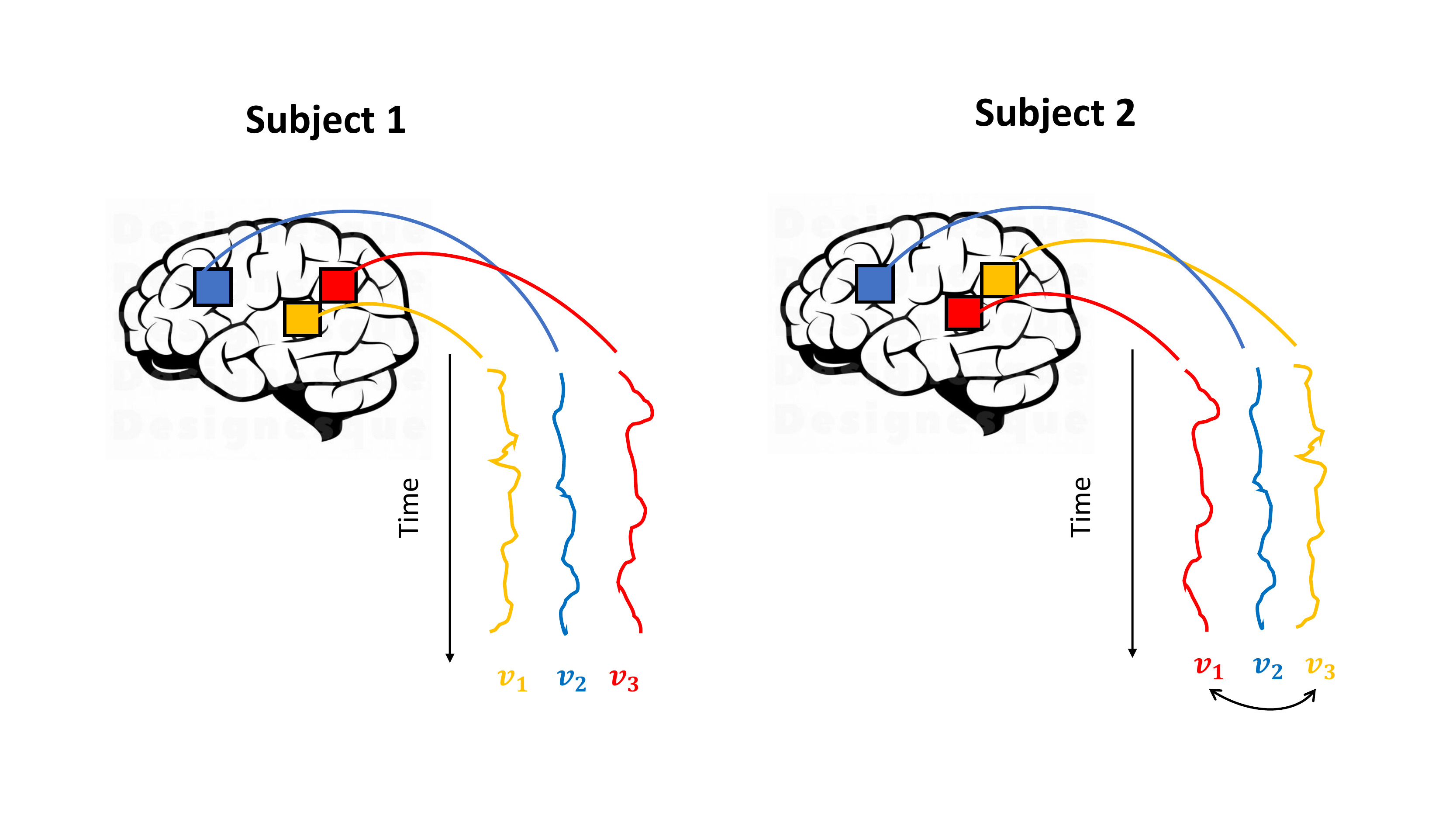}
	\caption{Illustration of functional misalignment between fMRI images, where three voxels’ time series are plotted considering two subjects. The time series of voxels $v_1$ and $v_3$ of the second subject are swapped with respect to the first subject.}
	\label{fig:toy}
\end{figure}

\subsection{Data Description}\label{inference}
We apply the proposed method to data from \cite{Pernet}, available at \url{https://openneuro.org/datasets/ds000158/versions/1.0.0}. The study consists of neural activations of $218$ subjects passively listening to vocal \gramm{(}i.e., speech\gramm{)} and non vocal sounds.\angela{Because the application has had a mere illustrative purpose, we choose to use a small number of subjects ($18$) to facilitate the example's reproducibility by the readers.}. We preprocessed the data using the \gramm{Functional MRI of the Brain} Software Library (FSL) \citep{FSL} using a standard processing procedure. For details about the experimental design and data acquisition, please see \cite{Pernet}. 

\subsection{Functional Connectivity}

We perform\gramm{ed} region of interest and seed-based correlation analysis \citep{cordes2000mapping}. \gramm{The s}eed-based correlation map shows the level of functional connectivity between a seed and every voxel in the brain, \gramm{whereas the} region of interest analysis expresses the functional correlation between predefined regions of interest coming from a standard atlas. \angela{The analysis process is defined as follows: \gramm{F}irst , the subject images are aligned using Algorithm \ref{algo2}, then the element-wise arithmetic mean across subjects is calculated, and finally, the functional connectivity analysis is developed on this average matrix.}

We take the frontal pole as seed, being a region with functional diversity \citep{liu2013connectivity}. \angela{The anatomical alignment considered here refers to the MNI space normalization \citep{Jenkinson}.} Figure \ref{fig:connectivitySeed} shows the correlation values between the seed and each voxel in the brain using data without functional alignment \gramm{(}top of Figure \ref{fig:connectivitySeed}\gramm{)} and with functional alignment using the Efficient ProMises model\gramm{(} bottom of Figure \ref{fig:connectivitySeed}\gramm{)}. 
The first evidence is that the functional alignment produces more interpretable maps, where the various regions, such as the superior temporal gyrus, are delineated by marked spatial edges, while the non-aligned map produces more spread regions, hence being less interpretable. 
It is interesting to evaluate the regions more correlated with the frontal pole, for example, the superior temporal gyrus. This region is associated with the processing of auditory stimuli. The correlation of the superior temporal gyrus with the seed is clear in the bottom part of Figure \ref{fig:connectivitySeed}, where functionally aligned images are used. 

\begin{figure}[!ht]
	\centering
	\minipage{.6\textwidth}%
	\includegraphics[width=\linewidth]{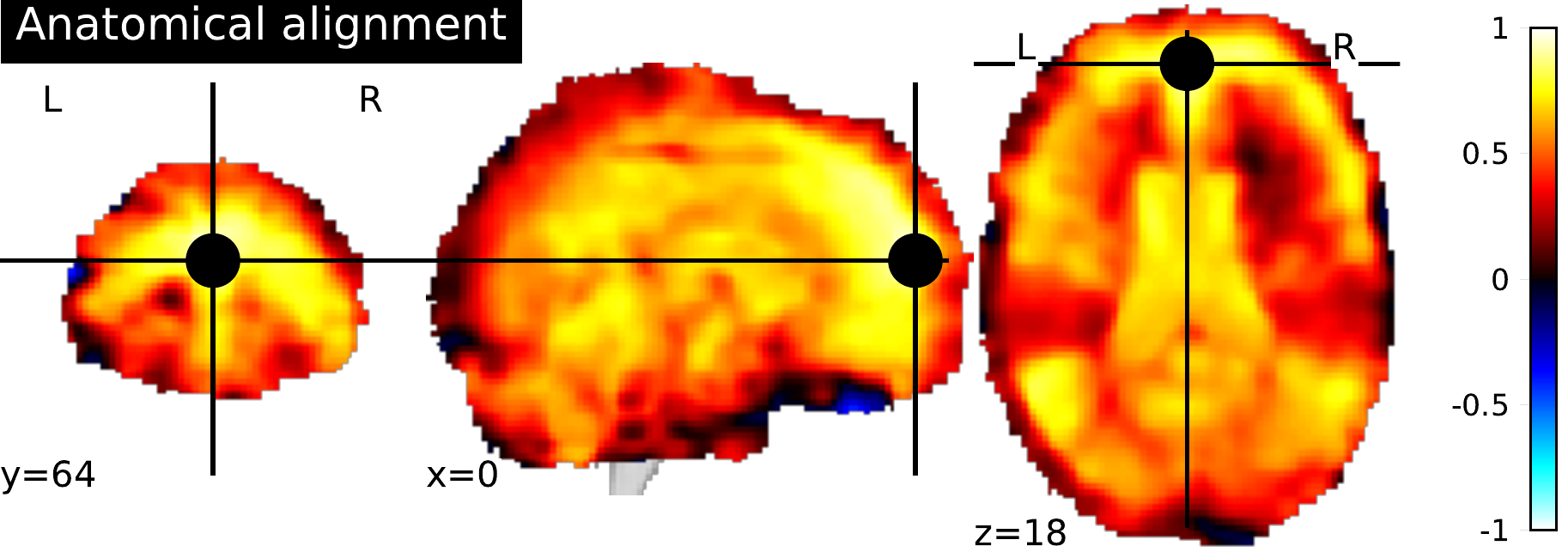}
	\endminipage\vfill
	\minipage{.6\textwidth}%
	\includegraphics[width=\linewidth]{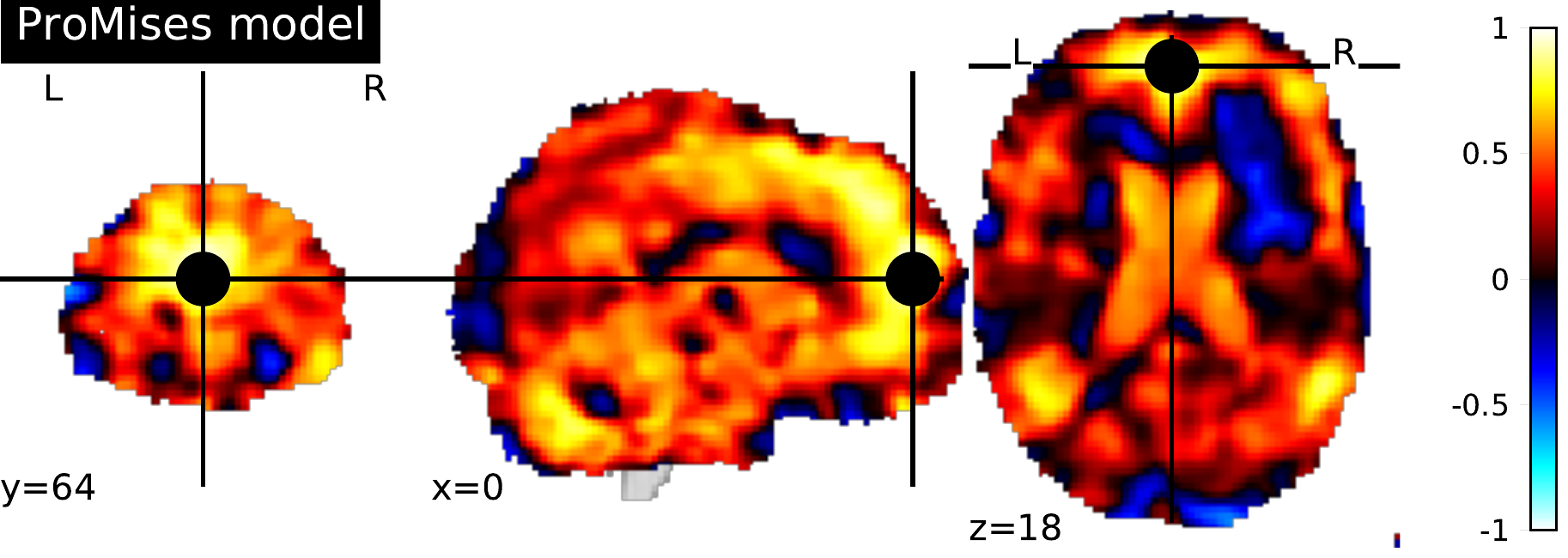}
	\endminipage
	\caption{Seed-based correlation map for $\boldsymbol{M}$, using data only aligned anatomically (top figure), and data also functionally aligned by the Efficient ProMises model (bottom figure). The black point refers to the seed used, i.e., frontal pole with MNI coordinates $(0,64, 18)$. So, the brain map indicates the level of correlation between each voxel and the frontal pole.}
	\label{fig:connectivitySeed}
\end{figure}

\gramm{In contrast}, the region of interest correlations analysis shows the integration mechanisms between specialized brain areas. Figure \ref{fig:connectivity} indicates the correlation matrices of time-series extracted from the $39$ main regions of the atlas of \cite{varoquaux2011multi}. Using functionally aligned data \gramm{(}right \gramm{side} of Figure \ref{fig:connectivity}\gramm{)} we can see delineated blocks of synchronized regions that can be interpreted as large-scale functional networks. Instead, using data without functional alignment \gramm{(}left \gramm{side} of Figure \ref{fig:connectivity}\gramm{)} the distinctions between blocks are clearly wors\gramm{e}. Using functionally aligned data, the left and right visual systems, composed of the dorsolateral prefrontal cortex (DLPFC), frontal pole (Front pol), and parietal (Par), are clearly visible\gramm{, whereas} in the analysis using functionally \gramm{non}aligned data, this distinction is hidden by noise.

\begin{figure}[!ht]
	\centering
	\minipage{.5\textwidth}%
	\includegraphics[width=\linewidth]{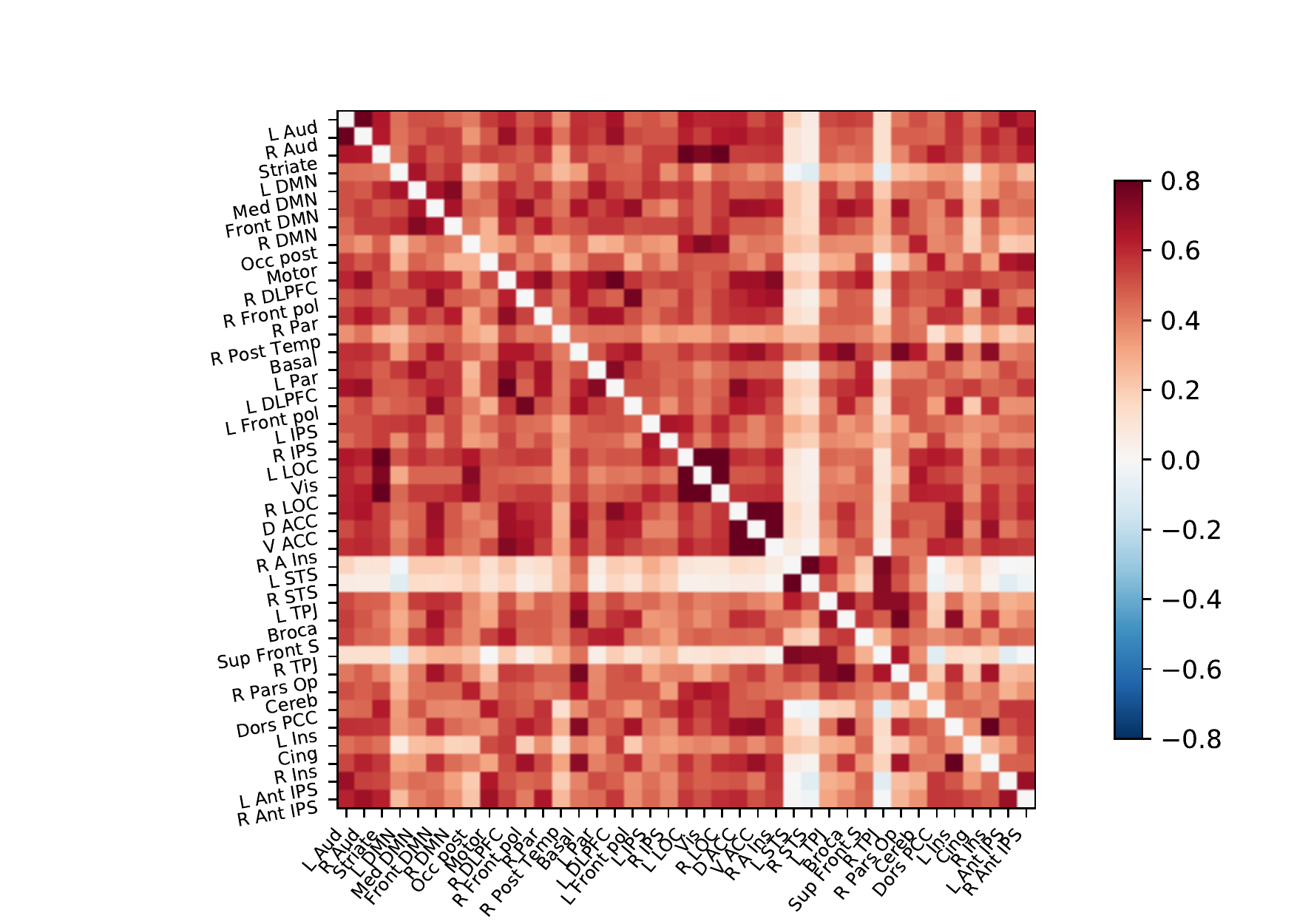}
	\endminipage\hfill
	\minipage{.5\textwidth}%
	\includegraphics[width=\linewidth]{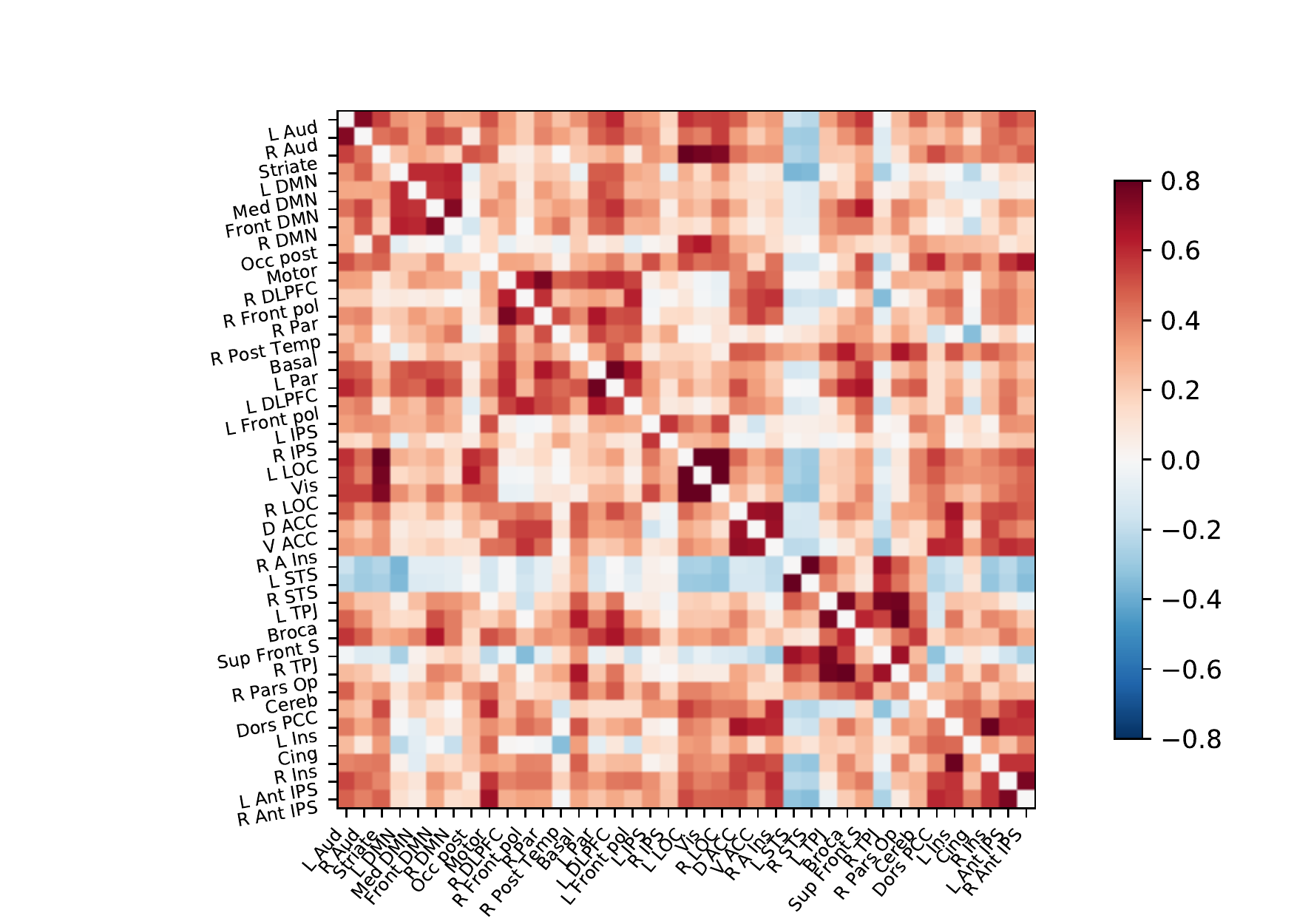}
	\endminipage
	\caption{Correlation matrix for $\boldsymbol{M}$, using data only aligned anatomically, left figure, and data also functionally aligned by the Efficient ProMises model, right figure. The cells of the matrix represent the correlation between the regions (represented by the row/column labels) of the \cite{varoquaux2011multi}'s atlas.}
	\label{fig:connectivity}
\end{figure}

\angela{The preprocessed data are available on the GitHub repository: \url{http://github.com/angeella/fMRIdata}, as well as the code used to perform functional connectivity:  \url{http://github.com/angeella/ProMisesModel/Code/Auditory}.}

\section{Discussion}\label{conclusion}
The ProMises model provides a methodologically grounded approach to the Procrustes problem allowing functional alignment on high-dimensional data in a computationally efficient way. The issues of the perturbation model \citep{Goodall} \gramm{---}non-uniqueness, critical interpretation, and \gramm{in}applicability when $n \ll m$\gramm{---} are completely surpassed thanks to our Bayesian extension. Indeed, the ProMises method returns unique and interpretable orthogonal transformations, and \gramm{its efficient} approach extends the applicability to high-dimensional data. The presented method is particularly useful in fMRI data analysis \gramm{because} it allows the functional alignment of images having roughly $200 \times 200,000$ dimensions, obtaining a unique representation of the aligned images in the brain space and a unique interpretation of the related results. 

\angela{In the application example presented in Section \ref{application}, a subsample was analyzed. However, the algorithm has a linear growth in $N$, and therefore, it is not a problem to work with larger samples. Also, the algorithm permits a parallel computation for the subjects.}

The Bayesian framework gives the user the advantage and the duty to insert prior information into the model through $k$ and $\boldsymbol{F}$. The parameter $k$ plays the role of regularization parameter, which is rarely known a priori. We estimated it by cross-validation, \gramm{although} it may be interesting to adapt approximations-based methods \gramm{(}e.g., generalized cross-validation\gramm{)} to reduce the computational burden. \angela{Alternatively, we could assume a prior distribution taking values in $\mathbb{R}^{+}$ for the regularization parameter $k$ and proceed to jointly estimate this parameter as well.} More interestingly, the matrix $\boldsymbol{F}$ addresses the estimate of the optimal rotations, which is favorable in the analysis of fMRI data \gramm{because}, in this context, the variables have a spatial anatomical location.
In the example in Section \ref{application} our definition of $\boldsymbol{F}$ favors the combination of voxels with equal location. However, a more thoughtful specification can entirely exploit the voxels' specific spatial position in the anatomical template. 
This opens up the possibility to explore various specifications of $\boldsymbol{F}$ and will be the subject of further research.


\section*{Acknowledgements}
The authors thank the Editor, the Associate Editor as well as the two anonymous referees for helpful comments that greatly improved the paper. Angela Andreella gratefully acknowledges funding from the grant BIRD2020/SCAR\_ASEGNIBIRD2020\_01 of the University of Padova, Italy, and PON 2014-2020/DM 1062 of the Ca' Foscari University of Venice, Italy. Some of the computational analyses done in this manuscript were carried out using the University of Padova Strategic Research Infrastructure Grant 2017: ``CAPRI: Calcolo ad Alte Prestazioni per la Ricerca e l’Innovazione'', \url{http://capri.dei.unipd.it}. The authors thank Prof. Umberto Castiello and Dr. Silvia Guerra for sharing the cinematic plant data.

\vspace{\fill}\pagebreak



 \newcommand{\noop}[1]{}


\pagebreak

\appendix
\section{Proof of Theorems and Lemmas}\label{Appendix}
\begingroup
\renewcommand\thetheorem{1}
\begin{theorem}[\cite{Douglas}]\label{theorem1}
	Consider the perturbation model described in Definition \ref{Procrustes_new}, and the singular value decomposition $\boldsymbol{X}_i^\top \boldsymbol{\Sigma}_n^{-1} \boldsymbol{M} \allowbreak \boldsymbol{\Sigma}_m^{-1}= \boldsymbol{U}_i \boldsymbol{D}_i \boldsymbol{V}_i^\top$. So, the maximum likelihood estimators equal $\hat{\boldsymbol{R}}_i = \boldsymbol{U}_i \boldsymbol{V}_i^\top$, and $\hat{\alpha_i}_{\boldsymbol{\hat{R}}_i} = ||\Sigma_m^{-1/2} \hat{\boldsymbol{R}}_i^\top \boldsymbol{X}_i^\top \boldsymbol{\Sigma_n}^{-1/2}||^2/tr(\boldsymbol{D}_i)$.
\end{theorem}

\begin{proof}\label{proof1}
	
	The proof comes directly from \cite{Douglas}, we report here the final part that we will use for other proofs. We can write Equation \eqref{eq:eq15_new} as:
	\begin{equation*}
	\dfrac{1}{\alpha_i} \boldsymbol{X}_i \boldsymbol{R}_i - \boldsymbol{M} = \boldsymbol{E}_i \sim \mathcal{MN}(0, \boldsymbol{\Sigma}_n, \boldsymbol{\Sigma}_m).
	\end{equation*}
	
	The log-likelihood for $\boldsymbol{R}_i$ equals:
	\begin{align*}
	\ell(\boldsymbol{R}_i) = - &\frac{1}{2}  \sum_{i = 1}^{N} tr \{ 
	\boldsymbol{\Sigma}_m^{-1}(\dfrac{1}{\alpha_i}  \boldsymbol{X}_i \boldsymbol{R}_i -  \boldsymbol{M})^\top \boldsymbol{\Sigma}_n^{-1}(\dfrac{1}{\alpha_i}  \boldsymbol{X}_i \boldsymbol{R}_i -  \boldsymbol{M})\} + C,
	\end{align*}
	
	where $(\boldsymbol{M}, \boldsymbol{\Sigma}_n, \boldsymbol{\Sigma}_m)$ are known nuisance parameter and $C$ is a constant value.
	So, we have
	
	\begin{align*}
	\ell(\boldsymbol{R}_i) 
	&= \dfrac{1}{\alpha}_i tr\{\boldsymbol{R}_i^\top \boldsymbol{X}_i^\top \boldsymbol{\Sigma}_n^{-1} \boldsymbol{M} \boldsymbol{\Sigma}_m^{-1}\}+ C^\star, 
	\end{align*}
	where $C^\star$ is a constant.
	The maximization of the log-likelihood function $\ell(\boldsymbol{R}_i)$ leads to
	
	\begin{align}
	\boldsymbol{\hat{R}}_i &= \arg \max_{\boldsymbol{R}_i \in \mathcal{O}(m)} <\boldsymbol{R}_i, \boldsymbol{X}_i^\top \boldsymbol{\Sigma}_n^{-1} \boldsymbol{M} \boldsymbol{\Sigma}_m^{-1}>\nonumber \\
	&= \arg \max_{\boldsymbol{R}_i \in \mathcal{O}(m)} <\boldsymbol{R}_i, \boldsymbol{U}_i \boldsymbol{D}_i \boldsymbol{V}_i>  =\arg\max_{\boldsymbol{R}_i\in \mathcal{O}(m)}<\boldsymbol{D}_i,\boldsymbol{U}_i^\top \boldsymbol{R}_i \boldsymbol{V}_i> \nonumber\\
	&= \arg\max_{\boldsymbol{R}_i \in \mathcal{O}(m)}<\boldsymbol{D}_i,\boldsymbol{R}_i^o> = \boldsymbol{U}_i\boldsymbol{V}_i^\top, \label{eq:eq71} 
	\end{align}
	where $\boldsymbol{R}_i^o = \boldsymbol{U}_i^\top \boldsymbol{R}_i \boldsymbol{V}_i \in \mathcal{O}(m)$.
	The step \eqref{eq:eq71} is proved by \citet{Gower}, i.e., $<\boldsymbol{D}_i,\boldsymbol{R}_i^o>$ is maximum when $\boldsymbol{R}_i^o= \boldsymbol{I}_m$, giving $\boldsymbol{I}_m = \boldsymbol{U}_i^\top \boldsymbol{\hat{R}}_i \boldsymbol{V}_i $. We can note that the maximum likelihood estimator $\hat{\boldsymbol{R}}_i$ does not depend on $\alpha_i$.
	
	Consider the profile log-likelihood for $\alpha_i$:
	\begin{align*}
	\ell_{p}(\alpha_i) &=- \frac{1}{2} \sum_{i = 1}^{N} tr\{ \boldsymbol{\Sigma}_m^{-1/2} (\dfrac{1}{\alpha_i} \boldsymbol{X}_i \hat{\boldsymbol{R}}_i - \boldsymbol{M})^\top \boldsymbol{\Sigma}_n^{-1/2} \boldsymbol{\Sigma}_n^{-1/2} (\dfrac{1}{\alpha_i} \boldsymbol{X}_i \hat{\boldsymbol{R}}_i - \boldsymbol{M} )\boldsymbol{\Sigma}_m^{-1/2} \} + S \\
	&= - \frac{1}{2} \{\sum_{i = 1}^{N} \dfrac{1}{\alpha_i^2} tr(\boldsymbol{\Sigma}_m^{-1/2} \hat{\boldsymbol{R}}_i^\top \boldsymbol{X}_i^\top \boldsymbol{\Sigma_n}^{-1/2} \boldsymbol{\Sigma}_n^{-1/2} \boldsymbol{X}_i \hat{\boldsymbol{R}}_i \boldsymbol{\Sigma}_m^{-1/2}) \\
	& \quad + tr(\boldsymbol{\Sigma}_m^{-1/2} \boldsymbol{M}^\top \boldsymbol{\Sigma}_n^{-1/2} \boldsymbol{\Sigma}_n^{-1/2}\boldsymbol{M}  \boldsymbol{\Sigma}_m^{-1/2}) \\
	& \quad - 2  \dfrac{1}{\alpha_i}    tr(\boldsymbol{\Sigma}_m^{-1/2} \hat{\boldsymbol{R}}_i^\top \boldsymbol{X}_i^\top \boldsymbol{\Sigma}_n^{-1/2} \boldsymbol{\Sigma}_n^{-1/2}\boldsymbol{M}  \boldsymbol{\Sigma}_m^{-1/2})\} + S \\
	&= \sum_{i = 1}^{N} - \dfrac{1}{2 \alpha_i^2} ||\boldsymbol{\Sigma}_m^{-1/2} \hat{\boldsymbol{R}}_i^\top \boldsymbol{X}_i^\top \boldsymbol{\Sigma}_n^{-1/2}||^2 + \dfrac{1}{\alpha_i} tr(\boldsymbol{D}_i) + S^\star,
	\end{align*}
	where $S$, and $S^\star$ are constant values. Taking the first derivative, we have:
	\begin{align*}
	\dfrac{\partial \ell_{p}(\alpha_i)}{\partial \alpha_i}&= \alpha_i^{-1}||\boldsymbol{\Sigma}_m^{-1/2} \hat{\boldsymbol{R}}_i^\top \boldsymbol{X}_i^\top \boldsymbol{\Sigma}_n^{-1/2}||^2 -  tr(\boldsymbol{D}_i) \\
	\hat{\alpha_i}_{\hat{R_i}} &= \frac{||\boldsymbol{\Sigma}_m^{-1/2} \hat{\boldsymbol{R}}_i^\top \boldsymbol{X}_i^\top \boldsymbol{\Sigma_n}^{-1/2}||^2}{tr(\boldsymbol{D}_i)},
	\end{align*}
	having $\hat{\boldsymbol{R}}_i = \boldsymbol{U}_i\boldsymbol{V}_i^\top$ and $\alpha_i \in {\rm I\!R}^{+}$.
	
\end{proof}

\begingroup
\renewcommand\thelemma{2}
\begin{lemma}
	Consider the perturbation model described in Definition \ref{Procrustes_new}, with $\boldsymbol{R}_i$ distributed accordingly to \eqref{fisher}, then the posterior distribution $f(\boldsymbol{R}_i| k, \boldsymbol{F}, \boldsymbol{X}_i)$ is conjugate distribution to the von Mises-Fisher prior distribution with location posterior parameter equals
	\begin{align*}
	\boldsymbol{F}^\star=\boldsymbol{X_i}^\top \boldsymbol{\Sigma}_n^{-1} \boldsymbol{M} \boldsymbol{\Sigma}_m^{-1} + k \boldsymbol{F}.
	\end{align*}
\end{lemma}

\begin{proof}
	The joint posterior distribution is defined as
	\begin{align}
	\prod_{i=1}^{N} f(\boldsymbol{R}_i| \boldsymbol{X}_i, \boldsymbol{M}, \boldsymbol{\Sigma}_m, \boldsymbol{\Sigma}_n, k, \boldsymbol{F}) =& \prod_{i=1}^{N} \exp[-\dfrac{1}{2}tr\{ \boldsymbol{\Sigma}_m^{-1/2} (\frac{1}{\alpha_i} \boldsymbol{X}_i \boldsymbol{R}_i - \boldsymbol{M})^\top \boldsymbol{\Sigma}_n^{-1/2} \boldsymbol{\Sigma}_n^{-1/2} \nonumber\\ 
	&(\frac{1}{\alpha_i} \boldsymbol{X_i} \boldsymbol{R_i} - \boldsymbol{M} )\boldsymbol{\Sigma}_m^{-1/2} \}] \nonumber 
	\\
	\cdot& \exp\big\{k tr( \boldsymbol{F}^\top \boldsymbol{R}_i)\big\} \cdot C \nonumber \\
	=&  \exp (-\sum_{i=1}^{N} \dfrac{1}{2} \Psi_i ) \nonumber\\
	\cdot &\exp( \sum_{i=1}^{N} < \boldsymbol{X}_i^\top \boldsymbol{\Sigma}_n^{-1} \boldsymbol{M} \boldsymbol{\Sigma}_m^{-1} + k  \boldsymbol{F}, \boldsymbol{R}_i> )\label{post},
	\end{align}
	where $\Psi_i = f(X_i)$ and $C$ is a constant value.
	The quantity \eqref{post} is a kernel of a matrix von Mises-Fisher distribution with location parameter equals 
	\begin{align*}
	\boldsymbol{F}^\star = \boldsymbol{X}_i^\top \boldsymbol{\Sigma}_n^{-1} \boldsymbol{M} \boldsymbol{\Sigma}_m^{-1} + k  \boldsymbol{F}.
	\end{align*}
\end{proof}
\begingroup
\renewcommand\thetheorem{2}
\begin{theorem}
	The ProMises model is defined as the perturbation model \eqref{eq:eq15_new} imposing the prior distribution \eqref{prior} for $\alpha_i \boldsymbol{R}_i$. Let the singular value decomposition of $\boldsymbol{X}_i^\top \boldsymbol{\Sigma}_n^{-1} \boldsymbol{M} \boldsymbol{\Sigma}_m^{-1} + k \boldsymbol{F}$ be $\boldsymbol{U}_i \boldsymbol{D}_i \boldsymbol{V}_i^\top$. Then the maximum a posteriori estimators equal $\boldsymbol{\hat{R}^{'}}_i = \boldsymbol{U}_i \boldsymbol{V}_i^\top$, and \\ $\hat{\alpha_i}_{\boldsymbol{\hat{R}^{'}}_i}^{'}=||\boldsymbol{\Sigma}_m^{-1/2} \boldsymbol{\hat{R}^{'}}_i \boldsymbol{X}_i^\top \boldsymbol{\Sigma}_n^{-1/2}||^2/tr(\boldsymbol{D}_i)$.
\end{theorem}
\begin{proof}
	Consider the same assumption of Theorem \ref{thm:R}, the log-posterior distribution for $\boldsymbol{R}_i$ and $\alpha_i$ equals:
	\begin{align*}
	\log f(\alpha_i, \boldsymbol{R_i} | \boldsymbol{X_i}, \boldsymbol{M}, \boldsymbol{\Sigma}_m, \boldsymbol{\Sigma}_n, k, \boldsymbol{F}) =& - \frac{1}{2} \sum_{i = 1}^{N} tr \{ \boldsymbol{\Sigma}_m^{-1} (\dfrac{1}{\alpha_i} \boldsymbol{X}_i \boldsymbol{R}_i - \boldsymbol{M})^\top \boldsymbol{\Sigma}_n^{-1} (\dfrac{1}{\alpha_i} \boldsymbol{X}_i \boldsymbol{R}_i - \boldsymbol{M}) \\
	& - 2 \dfrac{k}{\alpha_i} \boldsymbol{F}^\top \boldsymbol{R}_i \} - \log(\alpha_i) + K,
	\end{align*}
	where $K$ is a constant value. Following the same steps of Theorem \ref{thm:R}'s proof, the maximum a posteriori estimate equals
	\begin{align}
	\boldsymbol{\hat{R}^{'}}_i &= \arg \max_{\boldsymbol{R}_i \in \mathcal{O}(m)} <\boldsymbol{R}_i, \boldsymbol{X}_i^\top \boldsymbol{\Sigma}_n^{-1} \boldsymbol{M} \boldsymbol{\Sigma}_m^{-1}>  + k < \boldsymbol{R}_i, \boldsymbol{F}>\nonumber\\
	&= \arg \max_{\boldsymbol{R}_i \in \mathcal{O}(m)} <\boldsymbol{R}_i, \boldsymbol{X}_i^\top \boldsymbol{\Sigma}_n^{-1} \boldsymbol{M} \boldsymbol{\Sigma}_m^{-1} + k\boldsymbol{F}> \nonumber\\
	&= \arg \max_{\boldsymbol{R}_i \in \mathcal{O}(m)} <\boldsymbol{R}_i, \boldsymbol{U}_i \boldsymbol{D}_i \boldsymbol{V}_i> =\arg\max_{\boldsymbol{R}_i\in \mathcal{O}(m)}<\boldsymbol{D}_i,\boldsymbol{U}_i^\top \boldsymbol{R}_i \boldsymbol{V}_i> \nonumber\\
	&= \max_{\boldsymbol{R}_i \in \mathcal{O}(m)}<\boldsymbol{D}_i,\boldsymbol{R}_i^o> = \boldsymbol{U}_i\boldsymbol{V}_i^\top, \label{eq:eq72} 
	\end{align}
	where step \eqref{eq:eq72} is proved in the same way as step \eqref{eq:eq71} of Theorem \ref{theorem1}' proof.
	
	Then we compute the maximum a posteriori estimate for $\alpha_i$:
	\begin{align*}
	\hat{\alpha_i}_{\boldsymbol{\hat{R}^{'}}_i}^{'} &= \arg \max_{\alpha_i \in {\rm I\!R}^{+}} - \dfrac{1}{2 \alpha_i^2} ||\boldsymbol{\Sigma}_m^{-1/2} \boldsymbol{\hat{R}^{'\top}}_i \boldsymbol{X}_i^\top \boldsymbol{\Sigma_n}^{-1/2}||^2 + \dfrac{1}{\alpha_i} <\boldsymbol{\hat{R}^{'\top}}_i, \boldsymbol{X}_i^\top \boldsymbol{\Sigma}_n^{-1} \boldsymbol{M} \boldsymbol{\Sigma}_m^{-1}> \\ 
	&+ \dfrac{k}{\alpha_i} tr(\boldsymbol{F}^\top \boldsymbol{\hat{R}^{'}}_i) - \log(\alpha_i) + P,
	\end{align*}
	where $P$ is a constant value. Compute the first derivative and set it to zero:
	\begin{align*}
	\alpha_i^{-2}||\boldsymbol{\Sigma}_m^{-1/2} \boldsymbol{\hat{R}^{'\top}}_i \boldsymbol{X}_i^\top \boldsymbol{\Sigma}_n^{-1/2}||^2 - \alpha_i^{-1} <\boldsymbol{\hat{R}^{'}}_i, \boldsymbol{X}_i^\top \boldsymbol{\Sigma}_n^{-1} \boldsymbol{M} \boldsymbol{\Sigma}_m^{-1} + k \boldsymbol{F}> - 1 = 0,
	\end{align*}
	so, applying the Vi\'ete theorem \citep{Vieta}, $\hat{\alpha_i}_{\boldsymbol{\hat{R}^{'}}_i}$ equals:
	\begin{align*}
	\hat{\alpha_i}_{\boldsymbol{\hat{R}^{'}}_i} &= \frac{||\boldsymbol{\Sigma}_m^{-1/2} \boldsymbol{\hat{R}^{'\top}}_i \boldsymbol{X}_i^\top \boldsymbol{\Sigma}_n^{-1/2}||^2}{tr(\boldsymbol{D}_i)},
	\end{align*}
	under the condition $\dfrac{tr(\boldsymbol{D}_i)}{||\Sigma_m^{-1/2} \boldsymbol{\hat{R}^{'\top}}_i \boldsymbol{X}_i^\top \boldsymbol{\Sigma}_n^{-1/2}||^2} \gg \dfrac{1}{tr(\boldsymbol{D}_i)}$, and $\boldsymbol{D}_i$ coming from the singular value decomposition of $\boldsymbol{X}_i^\top \boldsymbol{\Sigma}_n^{-1} \boldsymbol{M} \boldsymbol{\Sigma}_m^{-1} + k \boldsymbol{F}$.
\end{proof}	

\renewcommand\thelemma{4}
\begin{lemma}
	Consider $\boldsymbol{X}_i \in {\rm I\!R}^{n \times m}$, if $n < m$, then the maximum likelihood estimate for $\boldsymbol{R}_i$ defined in Theorem \ref{thm:R} is not unique.
\end{lemma}

\begin{proof}
	In practice, without loss of generality, the Procrustes problem can be resumed as:
	\begin{align}
	\max_{\boldsymbol{R}_i \in \mathcal{O}(m)} tr(\boldsymbol{A}_i^\top \boldsymbol{R}_i),
	\label{eq:eq9}
	\end{align} 
	where $\boldsymbol{A}_i = \boldsymbol{X}_i^\top \boldsymbol{\Sigma}_n^{-1} \boldsymbol{M} \boldsymbol{\Sigma}_m^{-1}$. \citet{Trendafilov} and \citet[Lemma 1]{Myronenko} proved that the solution for \eqref{eq:eq9} is unique if and only if the matrix $\boldsymbol{A}_i$ has full rank. In Theorem \ref{thm:R} with $n < m$, $\boldsymbol{A}_i$ is equal to $\boldsymbol{X}_i^\top \boldsymbol{\Sigma}_n^{-1} \boldsymbol{M} \boldsymbol{\Sigma}_m^{-1}$ having rank lower than $m$, so the solution is not unique. Please refer to \citet{Trendafilov} and \citet[Lemma 1]{Myronenko} for further details about the complete proof.
\end{proof}
\begingroup
\renewcommand\thetheorem{3}
\begin{theorem}
	Consider the perturbation model in Definition \ref{Procrustes_new} with $\boldsymbol{\Sigma}_m = \sigma^2 \boldsymbol{I}_m$, and the thin singular value decompositions of $\boldsymbol{X}_i = \boldsymbol{L}_i \boldsymbol{S}_i \boldsymbol{Q}_i^\top$ for each $i = 1, \dots, N$, where $\boldsymbol{Q}_i$ has dimensions $n \times m$. The following holds
	\begin{align*}
	\max_{\boldsymbol{R}_i \in \mathcal{O}(m)} tr(\boldsymbol{R}_i^\top \boldsymbol{X}_i^\top \boldsymbol{\Sigma}_n^{-1} \boldsymbol{X}_j \boldsymbol{\Sigma}_m^{-1}) = \max_{\boldsymbol{R}_i^{\star} \in \mathcal{O}(n)} tr(\boldsymbol{R}_i^{ \star \top} \boldsymbol{Q}_i^\top \boldsymbol{X}_i^\top \boldsymbol{\Sigma}_n^{-1} \boldsymbol{X}_j \boldsymbol{\Sigma}_m^{-1} \boldsymbol{Q}_j^\top).
	\end{align*}
\end{theorem}
\begin{proof}

	Without loss of generality we consider $\boldsymbol{\Sigma}_m = \sigma^2\boldsymbol{I}_m$, and so the following objective function to maximize
	\begin{align*}
	tr(\boldsymbol{R}^\top_i \boldsymbol{X}^\top_i \boldsymbol{\Sigma}_n^{-1} \boldsymbol{X}_j \sigma^2 \boldsymbol{I}_m).
	\end{align*}
	We note that it is equivalent to maximize
	\begin{align}\label{block}
	tr(\boldsymbol{\Sigma}_n^{-1}\boldsymbol{X}_i \boldsymbol{R}_i \boldsymbol{X}_j^\top )
	\end{align}
	thanks to the trace's proprieties.
	
	Let consider the full singular value decomposition $\boldsymbol{X}_i = \boldsymbol{L}_i \boldsymbol{S}_i \boldsymbol{C}_i^\top$, where $\boldsymbol{S}_i \in {\rm I\!R}^{n \times m}$. The $\boldsymbol{S}_i$ matrix is defined as
	\begin{align*}
	\boldsymbol{S}_i = [ \boldsymbol{S}_i^\star  \,\, \boldsymbol{O}]  ,
	\end{align*}
	where $\boldsymbol{S}_i^\star \in {\rm I\!R}^{n \times n}$ and $\boldsymbol{O}$ is a matrix of zero with $n \times (m - n)$ dimensions, since $rank(\boldsymbol{X}_i) = n \,\, \forall i = 1, \dots, N$. Therefore, Expression \eqref{block} equals
	\begin{align*}
	tr(\boldsymbol{\Sigma}_n^{-1}\boldsymbol{X}_i \boldsymbol{R}_i \boldsymbol{X}_j^\top) =  tr(\boldsymbol{\Sigma}_n^{-1}\boldsymbol{L}_i \boldsymbol{S}_i \boldsymbol{C}_i^\top \boldsymbol{R}_i \boldsymbol{C}_j \boldsymbol{S}_j^\top \boldsymbol{L}_j^\top) = tr(\boldsymbol{\Sigma}_n^{-1}\boldsymbol{L}_i \boldsymbol{S}_i \boldsymbol{R}_i^{o} \boldsymbol{S}_j^\top \boldsymbol{L}_j^\top),
	\end{align*}
	where $\boldsymbol{R}_i^{o} = \boldsymbol{C}_i^\top \boldsymbol{R}_i \boldsymbol{C}_j \in \mathcal{O}(m)$ being a product of orthogonal matrices.
	
	Partitioning $\boldsymbol{R}_i^{o}$ in blocks, i.e.,
	\begin{equation}
	\begin{bmatrix}
	\boldsymbol{R}_{11i}^o     \boldsymbol{R}_{12i}^o\\
	\boldsymbol{R}_{21i}^o      \boldsymbol{R}_{22i}^o
	\end{bmatrix}
	\end{equation}
	where $\boldsymbol{R}_{11i}^o \in {\rm I\!R}^{n \times n}$, $\boldsymbol{R}_{12i}^o \in {\rm I\!R}^{n \times m-n}$, $\boldsymbol{R}_{21i}^o \in {\rm I\!R}^{m-n \times n}$, and $\boldsymbol{R}_{22i}^o \in {\rm I\!R}^{m-n \times m-n}$, we have:
	\begin{align}\label{lightProof1}
	\boldsymbol{\Sigma}_n^{-1}\boldsymbol{L}_i \boldsymbol{S}_i \boldsymbol{R}_i^{o} \boldsymbol{S}_j^\top \boldsymbol{L}_j^\top &=  \boldsymbol{\Sigma}_n^{-1}\boldsymbol{L}_i [ \boldsymbol{S}_i^\star  \,\, \boldsymbol{O}]         \begin{bmatrix}
	\boldsymbol{R}_{11}^o     \boldsymbol{R}_{12}^o\\
	\boldsymbol{R}_{21}^o      \boldsymbol{R}_{22}^o
	\end{bmatrix}     \begin{bmatrix}
	\boldsymbol{S}_j^{\top\star} \\
	\boldsymbol{O}^\top
	\end{bmatrix} \boldsymbol{L}_j^\top \\
	&= [ \boldsymbol{\Sigma}_n^{-1}\boldsymbol{L}_i \boldsymbol{S}_i^\star \,\, \boldsymbol{O}] \begin{bmatrix}
	\boldsymbol{R}_{11i}^o     \boldsymbol{R}_{12i}^o\\
	\boldsymbol{R}_{21i}^o      \boldsymbol{R}_{22i}^o
	\end{bmatrix}   \begin{bmatrix}
	\boldsymbol{S}_j^{\top\star} \boldsymbol{L}_j^\top  \nonumber\\
	\boldsymbol{O}^\top
	\end{bmatrix} \\
	&=  \boldsymbol{\Sigma}_n^{-1}\boldsymbol{L}_i \boldsymbol{S}_i^\star \boldsymbol{R}_{11i}^{o}\boldsymbol{S}_j^{\top\star} \boldsymbol{L}_j^\top.   \nonumber
	\end{align}
	
	Then, we have
	\begin{align*}
	\max_{\boldsymbol{R}_i \in \mathcal{O}(m)} tr(\boldsymbol{\Sigma}_n^{-1}\boldsymbol{X}_i \boldsymbol{R}_i \boldsymbol{X}_j^\top) &=   \max_{\boldsymbol{R}_{11i}^o \in \mathcal{O}(n)} tr(\boldsymbol{\Sigma}_n^{-1}\boldsymbol{L}_i \boldsymbol{S}_i^\star \boldsymbol{R}_{11i}^o \boldsymbol{S}_j^{\star\top} \boldsymbol{L}_j^\top) \\
	&= \max_{\boldsymbol{R}_i^{\star} \in \mathcal{O}(n)} tr(\boldsymbol{\Sigma}_n^{-1}\boldsymbol{X}_i\boldsymbol{Q}_i\boldsymbol{R}_i^{ \star} \boldsymbol{Q}_j^\top \boldsymbol{X}_j^\top ).
	\end{align*}
	
	The last equality is due to
	\begin{align*}
	\boldsymbol{\Sigma}_n^{-1}\boldsymbol{X}_i\boldsymbol{Q}_i\boldsymbol{R}_i^{ \star} \boldsymbol{Q}_j^\top \boldsymbol{X}_j^\top  =
	\boldsymbol{\Sigma}_n^{-1}\boldsymbol{L}_i\boldsymbol{S}_i^\star \boldsymbol{Q}_i^\top \boldsymbol{Q}_i \boldsymbol{R}_i^{ \star} \boldsymbol{Q}_j^\top \boldsymbol{Q}_j
	\boldsymbol{S}_j^\star 
	\boldsymbol{L}_j^\top
	= \boldsymbol{\Sigma}_n^{-1} \boldsymbol{L}_i \boldsymbol{S}_i^\star \boldsymbol{R}_{i}^\star \boldsymbol{S}_j^{\star\top} \boldsymbol{L}_j^\top.
	\end{align*}
	
	So essentially, only the first $n$ dimensions are used in maximizing \eqref{block}, if $n < m$ in all Procrustes-based problem. 
\end{proof}

\renewcommand\thelemma{5}
\begin{lemma}\label{light_vmp}
	Consider the assumptions of Theorem \ref{thm1}, then 
	\begin{align*}
	\max_{\boldsymbol{R}_i \in \mathcal{O}(m)} tr(\boldsymbol{R}_i^\top \boldsymbol{X}_i^\top \boldsymbol{\Sigma_n}^{-1} \boldsymbol{X}_j \boldsymbol{\Sigma}_m^{-1} + k \boldsymbol{F}) = \max_{\boldsymbol{R}_i^{\star} \in \mathcal{O}(n)} tr\{\boldsymbol{R}_i^{\star \top} (\boldsymbol{Q}_i^\top \boldsymbol{X}_i^\top \boldsymbol{\Sigma}_n^{-1} \boldsymbol{X}_j \boldsymbol{\Sigma}_m^{-1} \boldsymbol{Q}_j^\top + k \boldsymbol{F}^\star)\},
	\end{align*}
	where $\boldsymbol{F} \in {\rm I\!R}^{m \times m}$ and $\boldsymbol{F}^\star \in {\rm I\!R}^{n \times n}$.
\end{lemma}

\begin{proof}
	Without loss of generality we consider $\boldsymbol{\Sigma}_m = \sigma^2\boldsymbol{I}_m$, and $\boldsymbol{F}^\star = \boldsymbol{Q}^\top_i \boldsymbol{F} \boldsymbol{Q}_j$. So, the following objective function to maximize
	
	\begin{align}\label{lightProof}
	\max_{\boldsymbol{R}_i^\star \in \mathcal{O}(n)} tr(\boldsymbol{R}_i^{\top\star} \boldsymbol{Q}_i^\top\boldsymbol{X}^\top_i \boldsymbol{\Sigma}_n^{-1} \boldsymbol{X}_j \boldsymbol{Q}_j) + k tr( \boldsymbol{R}_i^{\top\star} \boldsymbol{Q}_i^\top \boldsymbol{F} \boldsymbol{Q}_j).
	\end{align}
	
	The left part of the maximization \eqref{lightProof} equals \eqref{lightProof1}, while $ktr(\boldsymbol{R}_i^{\top\star} \boldsymbol{Q}_i^\top \boldsymbol{F} \boldsymbol{Q}_j)$ in $\mathcal{O}(n)$ is equivalent to $ktr( \boldsymbol{R}_i^{\top}  \boldsymbol{F})$ in $\mathcal{O}(m)$. 
	
\end{proof}




\end{document}